\documentclass[12pt]{article}

\usepackage{amssymb, amsmath, amsthm}
\usepackage{graphicx,subfig,tikz}
\usepackage[left=2cm,right=2cm,top=2cm,bottom=2cm]{geometry}

\newtheorem{theorem}{Theorem}
\newtheorem{lemma}{Lemma}

\usepackage[colorlinks=true,linktocpage=true,allcolors=blue]{hyperref}

\def\be{\begin{equation}}
\def\ee{\end{equation}}
\def\bea{\begin{eqnarray}}
\def\eea{\end{eqnarray}}

\providecommand{\d}{\mathop{}\!\mathrm{d}}
\renewcommand{\d}{\mathop{}\!\mathrm{d}}
\newcommand{\td}{\d}
\DeclareMathOperator{\sgn}{sgn}
\newcommand*{\omh}{\hat{\omega}}

\newcommand*{\dt}{\d t}
\newcommand*{\dr}{\d r}
\newcommand*{\dps}{\d\psi}
\newcommand*{\dth}{\d\theta}
\newcommand*{\dph}{\d \phi}
\newcommand*{\dv}{\d v}
\newcommand*{\ds}{\d s}
\newcommand*{\Ht}{\tilde{H}}
\newcommand*{\Kt}{\tilde{K}}
\newcommand*{\Lt}{\tilde{L}}
\newcommand*{\Mt}{\tilde{M}}

\newcommand*{\chiph}{\chi} 
\newcommand*{\omhph}{\omh}

\newcommand{\llangle}{\langle\!\langle}
\newcommand{\rrangle}{\rangle\!\rangle}

\newcommand*{\Ord}{%\mathcal{
O%}
}

\title{ \bf{Moduli space of supersymmetric solitons and black holes in
    five dimensions}}

\author{Veronika Breunh\"older\footnote{v.breunhoelder@ed.ac.uk} \\
  and \\ James Lucietti\footnote{j.lucietti@ed.ac.uk } \\ \\
  \small \sl School of Mathematics and Maxwell Institute of
  Mathematical Sciences, \\ \small \sl University of Edinburgh, The
  King's Buildings, \\ \small \sl Edinburgh, EH9 3FD, UK }

\date{}

\begin{document}

\maketitle

%\vskip1.5cm

\begin{abstract}
  We determine all asymptotically flat, supersymmetric and
  biaxisymmetric soliton and black hole solutions to five-dimensional minimal
  supergravity. In particular, we show that the solution must be a
  multi-centred solution with a Gibbons--Hawking base. The proof
  involves combining local constraints from supersymmetry with global
  constraints for stationary and biaxisymmetric
  spacetimes. This reveals that the horizon topology must be one of
  $S^3$, $S^1\times S^2$ or a lens space $L(p,1)$, thereby providing a
  refinement of the allowed horizon topologies.  We construct the general smooth solution for each possible rod structure. We find a
  large moduli space of black hole spacetimes with noncontractible 2-cycles for each
  of the allowed horizon topologies. In the absence of a black hole we
  obtain a classification of the known  `bubbling' soliton spacetimes.
\end{abstract}

\newpage

\tableofcontents

%\newpage
\section{Introduction}

The classification of isolated gravitating equilibrium states is a
problem of central importance in general relativity.  In
four-dimensional Einstein--Maxwell theory this question has been
answered, under some assumptions. Any asymptotically flat stationary
spacetime must contain a black hole (no soliton theorem), and
furthermore, the black hole uniqueness theorem implies it must be a
Kerr--Newman solution, see e.g.~\cite{Chrusciel:2012jk}.

For higher-dimensional general relativity, an analogous classification
of equilibrium states is a major open
problem~\cite{Emparan:2008eg}. It is known that both the no-soliton
and the black hole uniqueness theorems are violated, even within the
class of asymptotically flat (Minkowski $\mathbb{R}^{1,D-1}$)
spacetimes. The failure of the uniqueness theorem was first revealed
by the discovery of the black ring, an asymptotically flat stationary
black hole vacuum solution with a horizon of spatial topology
$S^1\times S^2$ \cite{Emparan:2008eg}. It is now expected that the
moduli space of stationary black hole solutions in higher dimensions
is very rich, although further explicit solutions are hard to come by.

The failure of the no-soliton theorem became apparent after the
construction of the `bubbling' spacetimes in
supergravity~\cite{Bena:2007kg} (see~\cite{Gibbons:2013tqa} for a
discussion of this).  In particular, there exist finite energy,
asymptotically flat, stationary spacetimes which are regular
everywhere and contain no black hole region. The simplest examples are
supersymmetric solutions to five-dimensional minimal supergravity
(Einstein--Maxwell theory with a Chern--Simons coupling). Such
spacetimes are topologically nontrivial and contain noncontracible 2-cycles, or
`bubbles', supported by magnetic flux. Indeed, such soliton spacetimes
do not exist in vacuum gravity.

The existence of bubbling spacetimes leads to a further violation of
black hole uniqueness. This is because one can envisage a black hole
sitting in a bubbling spacetime. Indeed, the first law of black hole
mechanics is modified by flux terms that couple to the
bubbles~\cite{Kunduri:2013vka}.  Supersymmetric solutions describing a
spherical black hole in an asymptotically flat bubbling spacetime can
be constructed explicitly~\cite{Kunduri:2014iga}.  Interestingly, this
leads to an entropy enigma which raises questions for the microscopic
description of black holes in string
theory~\cite{Horowitz:2017fyg}. Furthermore, these techniques have led
to the construction of the first example of an asymptotically flat
black hole with lens space topology, namely $L(2,1)$, termed a black
lens~\cite{Kunduri:2014kja}. Subsequently, black lenses with more
general horizon topology $L(p,1)$ have been
constructed~\cite{Tomizawa:2016kjh}. Thus, even in five-dimensional
spacetimes, the moduli space of black hole solutions is now expected
to be large and complicated.

A number of results have been derived that help constrain the
topology and symmetry of asymptotically flat black hole spacetimes,
for a review see~\cite{Hollands:2012xy}. Topological censorship
implies the domain of outer communication (DOC) must be simply
connected~\cite{Friedman:1993ty}. The horizon topology theorem states
that spatial sections of the horizon admit positive scalar curvature,
which in five dimensions only allows $S^3$, $S^1\times S^2$,
$S^3/\Gamma$ and connected sums thereof~\cite{Galloway:2005mf}.  The rigidity
theorem implies that a stationary and rotating black hole must also be
axisymmetric and thus possess an isometry group
$\mathbb{R} \times U(1)$~\cite{Hollands:2006rj}.  Of course, these are
all necessary conditions which must be satisfied; what is unclear is
whether black hole solutions to Einstein's equations which realise all
the above topology and symmetry constraints actually exist.  The key
question is: what is the moduli space of black hole solutions with a
given topology and/or symmetry? There are essentially no results which
address this question, except for static black hole solutions to Einstein--Maxwell theory for
which a uniqueness theorem has been
established~\cite{Gibbons:2002bh,Gibbons:2002ju,Gibbons:2002av,Kunduri:2017htl}.

In fact, the known explicit solutions possess more rotational symmetry
than that guaranteed by the rigidity theorem. In particular, the
five-dimensional solutions possess $\mathbb{R} \times U(1)^2$-symmetry
and therefore belong to the class of generalised Weyl
solutions~\cite{Emparan:2001wk, Harmark:2004rm}. As for
four-dimensional stationary and axisymmetric spacetimes the Einstein
equations become integrable and solutions can be classified in terms
of a `rod structure'. The rod structure is essentially a specification
of how the $U(1)^2$-action degenerates on the axes of symmetry and
given this data one can determine the spacetime and horizon
topology. Indeed, by exploiting this structure a uniqueness theorem,
which generalises the four-dimensional one, has been proven. This
states that five-dimensional, asymptotically flat, stationary and
biaxisymmetric vacuum black hole solutions are uniquely specified by
their mass, angular momenta and rod structure~\cite{Hollands:2007aj,
  Hollands:2008fm}. This has been generalised to Einstein--Maxwell
theory and minimal supergravity, where one finds that the magnetic
flux on every noncontratible 2-cycle must also be
specified~\cite{Hollands:2007qf, Tomizawa:2009tb, Armas:2009dd,
  Yazadjiev:2011fg, Armas:2014gga}.\footnote{Analogous results for
  asymptotically Kaluza--Klein solutions have also been
  obtained~\cite{Yazadjiev:2009nm, Yazadjiev:2010uu} (see
  also~\cite{Haas:2015xmc, Haas:2017kdp}).} However, in contrast to
the four-dimensional case, the rod structure in five dimensions can be
arbitrarily complicated in principle.  What is not understood, is the
existence problem: for what rod structures do suitably regular black
hole solutions actually exist? Therefore, even the classification of
five-dimensional stationary black holes with biaxial symmetry remains
open (although see~\cite{Khuri:2017xsc} for recent progress).

Of course, there are other types of symmetry assumptions which help
simplify the construction of solutions. The classification of
supersymmetric solutions has been well studied. Most of these works
consist of local classifications. That is, deriving local constraints
on the geometry arising from the existence of a suitable Killing
spinor. However, what has been largely unstudied is a global analysis
of supersymmetric solutions.  Thus, a natural question presents
itself: Can we classify all supersymmetric soliton and black hole
solutions in five dimensions? Clearly, this requires a global analysis
of suitably general supersymmetric solutions.  Previously, a
uniqueness theorem for supersymmetric spherical topology black holes
in five-dimensional supergravity was
proven~\cite{Reall:2002bh,Gutowski:2004bj}, showing that the only
solution was the BMPV black hole \cite{Breckenridge:1996is}. However, due to an overly
restrictive assumption (the supersymmetric Killing field is
strictly timelike in the DOC) it excluded the recently constructed black holes in
bubbling spacetimes and black lenses~\cite{Kunduri:2014iga,Kunduri:2014kja,Tomizawa:2016kjh}. The analogous result in four
dimensions shows that the Majumdar--Papapetrou solutions are the most
general supersymmetric black holes in Einstein--Maxwell
theory~\cite{Chrusciel:2005ve}.

In fact, the recent new examples of supersymmetric black hole
solutions to five-dimensional supergravity all possess an
$\mathbb{R} \times U(1)^2$-symmetry and hence are in the class of Weyl
solutions (coupled to Maxwell field). Indeed, one
can assign them a rod structure which thus demonstrates that solutions
with more nontrivial rod structures do indeed exist in this case.
Therefore, an even simpler question presents itself: can we classify
all supersymmetric and biaxisymmetric soliton and black hole solutions
in five dimensions?  The purpose of this paper it is to show this is
indeed possible.

We will work in the simplest theory where such solutions exist,
namely five-dimensional minimal supergravity. Supersymmetric solutions to this
theory have been extensively
studied~\cite{Gauntlett:2002nw}. Generically, they are timelike
fibrations over a hyper-K\"ahler base. Remarkably, it was shown that if
the base is a Gibbons--Hawking space, the full local form of the
solution can be determined in terms of harmonic functions on an
auxiliary $\mathbb{R}^3$. The known supersymmetric black hole solutions
(including black rings~\cite{Elvang:2004rt,Gauntlett:2004wh} and black
lenses) and soliton solutions~\cite{Bena:2005va,Berglund:2005vb} belong to this class and are constructed from harmonic functions of
multi-centred type. 

One of our main results is the following classification theorem, the
complete statement of which is given in Theorem \ref{classthm}.

\begin{theorem} Consider an asymptotically flat, supersymmetric and
  biaxisymmetric solution to five-dimensional minimal supergravity
  with a globally hyperbolic domain of outer communication,
  possibly containing a black hole. Then, the solution must have
  a Gibbons--Hawking base and the associated harmonic functions are of
  multi-centred type.
\end{theorem}

In the absence of a black hole, the above provides a classification of
the bubbling soliton spacetimes in this symmetry class. This appears
to be the first classification theorem known for such spacetimes. In
the black hole case we find a rich moduli space of solutions,
corresponding to bubbling spacetimes containing spherical black holes,
black rings, or black lenses.

The proof consists of combining the local constraints from
supersymmetry with global constraints for stationary and
biaxisymmetric solutions. The main structure of the proof is as
follows. As noted above supersymmetry determines the local form of the
solution in terms of a set of harmonic functions on an auxiliary
$\mathbb{R}^3$~\cite{Gauntlett:2002nw}. This provides a key
simplification which is not available in vacuum gravity.  Thus, the
proof reduces to a global analysis of this class of solutions.  The
structure of the orbit space of the domain of outer communication of Weyl solutions is that of a 2d
manifold, with a boundary which corresponds to horizons (if there
is a black hole) or axes on which certain linear combinations of the
biaxial Killing fields vanish, and corners where both biaxial
Killing fields vanish~\cite{Hollands:2007aj, Hollands:2008fm}.  Using
the known classification of near-horizon
geometries~\cite{Reall:2002bh} allows us to prove that a smooth
horizon corresponds to an isolated point on the boundary of the orbit
space and furthermore shows that harmonic functions possess (at most)
a simple pole at the horizon.  Requiring smoothness of the DOC and the
axes, together with some global constraints~\cite{Chrusciel:2008rh},
also shows that the harmonic functions are non-singular in the
interior of the orbit space and everywhere else on its boundary,
except at the corners where they possess simple poles.  Thus, the
number of simple poles the harmonic functions may possess is given by
the number of horizons plus the number of corners of the orbit space.

We will also perform a detailed analysis of the possible rod
structures and show that they are constrained by supersymmetry. In
particular, this constrains the allowed horizon topologies to be one
of $S^3$, $S^1\times S^2$ or a lens space $L(p,1)$. Interestingly, this
provides a refinement of the topologies allowed by biaxial symmetry;
in particular it rules out $L(p,q)$ with $q \neq 1$ (mod $p$).
Nevertheless, an infinite number of possible rod structures
remains. We construct the explicit solution for every rod structure
and determine the set of conditions required for the solution to be
smooth and causal (on and outside a horizon). This reveals a very
large moduli space of solutions both with and without a black hole and we give a general formula for the dimension of the moduli space.
We find that for $n$-centred solutions, the number of inequivalent rod
structures grows with $n$ and provide a formula for counting these for
each horizon topology. 

The organisation of this paper is as follows. In section
\ref{sec:susy} we examine local and global constraints on
supersymmetric solutions imposed by the existence of biaxial symmetry commuting with supersymmetry.  In section \ref{sec:hor} and \ref{sec:axes} we examine the constraints imposed by the existence of a smooth event horizon and axes of symmetry, respectively.  In section \ref{sec:modspace} we present our main classification theorem and examine the moduli space of solutions. We end with a discussion in section \ref{sec:dis}.

\section{Supersymmetric solutions in five dimensions} 
\label{sec:susy}

The bosonic sector of five-dimensional minimal supergravity consists of a spacetime metric $g$ and Maxwell field $F$ and the field equations are those for Einstein--Maxwell theory coupled to a Chern--Simons term. The general form of supersymmetric solutions of ungauged minimal
supergravity is well understood~\cite{Gauntlett:2002nw}.  Given a Killing spinor one can construct a smooth function $f$ and a Killing vector $V$, each quadratic in the spinor, such that $V \cdot V = -f^2$.  Thus $V$ must be nonspacelike so the classification divides into solutions where $V$ is either null or timelike (at least in some region). The solutions where $V$ is null can be fully determined and correspond to plane wave and pp-wave spacetimes. We will be interested in asymptotically flat solutions, possibly containing a black hole, which must be in the timelike class. 

In any
region where the supersymmetric Killing field $V$ is timelike the spacetime metric takes the general form
\begin{equation}
\td s^2 = -f^2( \td t + \omega)^2 + f^{-1} h \, ,  \label{g}
\end{equation} 
where $V = \partial_t$, $h$ is a hyper-K\"ahler metric on the
orthogonal space $B$ and $\omega$ is a 1-form on
$B$.  As we explain later, under the additional assumption of biaxial
symmetry, the base must be a Gibbons--Hawking (GH) space.  It is thus convenient to first consider solutions with a GH base.

\subsection{Gibbons--Hawking base}
\label{sec:GH}

In this section we take the base metric to be a Gibbons--Hawking space,
however we will not assume biaxial symmetry, so our analysis is valid
in the general class of GH solutions. The GH metric is
\begin{equation}
h = H^{-1} (\td \psi + \chi_i \td x^i)^2 + H \td x^i \td x^i \, ,  \label{GH}
\end{equation} 
where $x^i$, $i=1,2,3$, are Cartesian coordinates on $\mathbb{R}^3$, the
function $H$ is harmonic on $\mathbb{R}^3$, $\chi$ is a 1-form on
$\mathbb{R}^3$ satisfying 
\begin{equation}
\star_3 \td \chi = \td H\label{dchi}
\end{equation}
and $\partial_\psi$ is the triholomorphic Killing field.

As is well known, the {\it local} form of the supersymmetric solution
can then be completely determined under the assumption that the full
solution is invariant under the triholomorphic Killing
field~\cite{Gauntlett:2002nw}. Such solutions are then specified by 4
harmonic functions $H,K,L,M$, in terms of which
\begin{align}
  f^{-1} &= H^{-1} K^2 + L \, , \label{Ldef} \\
  \omega &= \omega_\psi(\td \psi + \chi_i \td x^i) + \hat{\omega}_i
           \td x^i \, ,
\end{align}
 where 
\begin{align}
  \omega_\psi &= H^{-2} K^3 + \frac{3}{2}H^{-1}KL + M \,
                , \label{Mdef} \\
  \star_3 \td \hat\omega &= H \td M - M \td H + \frac{3}{2}(K \td L -
                           L \td K) \, .\label{domegahat}
\end{align} 
The Maxwell field is then
 \begin{equation}
 \label{max}
 F =dA = \frac{\sqrt{3}}{2} \td \left[ f( \td t + \omega) - K H^{-1} (\td \psi+ \chi_i \td x^i)  - \xi_i \td x^i\right]  \, ,
 \end{equation}
where the 1-form $\xi$ satisfies  
\be
\label{dxi}
\star_3 \td \xi = - \td K.
\ee

We wish to perform a global analysis of this family of local metrics. To this end, it will be useful to record the spacetime invariants
\begin{equation}\label{invariants}
  \begin{aligned}
  V \cdot V&=  g_{tt} = -f^2= - \frac{H^2}{(K^2 +HL)^2}\,, \\
  \partial_\psi \cdot \partial_\psi &=  g_{\psi\psi} = \frac{1}{fH} -
  f^2 \omega_\psi^2 = -\frac{4 H^2 M^2
    + 12 H K L M - 4 H L^3 + 8K^3 M - 3 K^2 L^2 }{4(HL +
    K^2)^2} \,, \\
  V \cdot \partial_\psi&=   g_{t\psi} = - f^2 \omega_\psi =- \frac{
    H^2 M + \tfrac{3}{2} HKL +  K^3}{(K^2 + HL)^2}\,,  \\
  A_t &= \frac{\sqrt{3}}{2} f \,, \qquad
  A_\psi = \frac{\sqrt{3}}{2}\left( f \omega_\psi - \frac{K}{H}\right)
  = \frac{\sqrt{3}}{2} \frac{HM + \tfrac{1}{2}KL}{HL + K^2}\,.
  \end{aligned}
\end{equation}
One can see that $A_t, A_\psi$ are invariants as follows.\footnote{Smoothness of $A_t$ also follows from the fact that $f$ is a spacetime invariant (a bilinear in the Killing spinor).}  First note $\mathcal{L}_V F= \mathcal{L}_{\partial_\psi} F=0$ imply $\td \iota_V F=\td \iota_{\partial_\psi} F=0$.  Therefore, for a simply connected spacetime (as we will be interested in) we deduce the existence of two globally defined functions $\Phi, \Psi$ satisfying
\begin{equation}
\iota_V F = \frac{\sqrt{3}}{2}\td \Phi, \qquad \iota_{\partial_\psi} F =  \frac{\sqrt{3}}{2}\td \Psi\,.    \label{PhiPsi}
\end{equation}
These functions $\Phi,\Psi$ are the electric potential and a magnetic potential respectively. From (\ref{max}) we can identify these potentials up to an additive constant as $A_t=  - \frac{\sqrt{3}}{2}\Phi$ and $A_\psi = - \frac{\sqrt{3}}{2}\Psi$, establishing these components of the gauge field are indeed spacetime invariants. Thus
\begin{equation}
\Phi= -f, \qquad \Psi =- f \omega_\psi  + KH^{-1}   \label{potentials}
\end{equation}
(note the former is true even without a GH base).  % It immediately follows that $f$ is a smooth spacetime function (this improves upon knowing $V \cdot V= - f^2$ is an invariant). 
In terms of these invariants the solution is
\begin{equation}\label{solution}
  \begin{aligned}
    \td s^2 &= g_{tt} (\td t +\hat{\omega}_i \td x^i)^2 + 2 g_{t\psi} (\td
    t+\hat{\omega}_i \td x^i) ( \td \psi+ \chi_i \td x^i) + g_{\psi\psi} (\td \psi+
    \chi_i \td x^i)^2+  \frac{H}{f} \td x^i \td x^i \,,\\
    A &= A_t (\td t +\hat{\omega}_i \td x^i) +A_\psi (\td \psi+\chi_i \td x^i) - \xi_i \td x^i \,.
  \end{aligned}
\end{equation}
The inverse metric can be written as
\begin{equation}
\label{invg}
  \begin{aligned}
    g^{tt} &= - \frac{H}{f} g_{\psi\psi}
    + \frac{f}{H} \hat \omega_i \hat{\omega}_i\,,  \quad
    &g^{t i} &= - \frac{f}{H} \hat{\omega}_i \,,\\
    g^{t\psi} &=  \frac{H}{f} g_{t\psi} + \frac{f}{H} \hat{\omega}_i \chi_i \,,
    &g^{\psi i} &= - \frac{f}{H} \chi_i\,, \\
    g^{\psi\psi} &= -\frac{H}{f} g_{tt}  + \frac{f}{H} \chi_i \chi_i \,,
    &g^{ij} &= \frac{f}{H} \delta_{ij}\,,
\end{aligned}
\end{equation}
and the determinant of the metric is
\begin{equation}
\sqrt{- \det g} = \frac{H}{f}  = K^2+HL\,.
\end{equation}
We now provide a spacetime interpretation of the harmonic
functions. 

First define the determinant of the matrix of inner products of the Killing fields $\partial_t, \partial_\psi$,
\begin{equation}\label{I}
  N \equiv - \begin{vmatrix} g_{tt} & g_{t\psi} \\ g_{t\psi}
    & g_{\psi\psi} \end{vmatrix}\,.
\end{equation}
This will be a key invariant in our analysis. From \eqref{invariants} it follows that
\begin{equation}\label{H}
N = \frac{f}{H}  \implies H= \frac{f}{N}\,.
\end{equation}
Next, we relate the harmonic functions $K,L,M$ to invariants as
follows. From the above,
\begin{equation}\label{K}
  K H^{-1} =  \Psi - \frac{g_{t\psi}}{f} \implies K = \frac{1}{N} \left(
  f \Psi-  g_{t\psi} \right)  \,.
\end{equation}
Using \eqref{Ldef} and \eqref{Mdef}, together with the above
expression for $K$ gives, after a little algebra,
\begin{align}
  L  &= \frac{1}{N} \left(f g_{\psi\psi} + 2 g_{t\psi} \Psi -  f
       \Psi^2 \right)\,,  \label{L} \\
  M &= \frac{1}{2 N}  \left( g_{\psi\psi} g_{t\psi} - 3 f \Psi
      g_{\psi\psi}- 3 \Psi^2 g_{t\psi} + f \Psi^3 \right)\,.  \label{M}
\end{align}
Observe that $K$ is only defined up to 
\begin{equation}
K \to K + c H \,, \label{Kshift}
\end{equation}
where $c$ is a constant corresponding to the integration constant for
$\Psi$. It follows that $L$, $M$ are defined up to the shifts
\begin{equation}
  L \to L - 2 c K - c^2 H \,, \qquad M \to  M- \tfrac{3}{2} c L +
  \tfrac{3}{2} c^2 K + \tfrac{1}{2} c^3 H \,. \label{LMshift}
\end{equation}
The following result will be useful for a global analysis of solutions
with a GH base.

\begin{lemma} \label{GHlemma} Let $(g,F)$ be a supersymmetric solution
  to minimal supergravity with a Gibbons--Hawking base, $H,K,L,M$
  the associated harmonic functions defined up to
  {\normalfont(\ref{Kshift},\,\ref{LMshift})}, and let $N$ be defined
  as in \eqref{I}.

  \begin{enumerate}
  \item If $H,K,L,M$ are smooth and
    \begin{equation}
      K^2+HL>0 \label{smooth1}\,,
    \end{equation}
    then $(g,F)$ is smooth, $g^{-1}$ exists and is smooth, and $N>0$.
  \item Conversely, if $(g,F)$ is smooth and $N>0$, then $H,K,L,M$ are
    smooth and obey \eqref{smooth1}.
  \end{enumerate}
\end{lemma}

\begin{proof}
  Smoothness of $H,K,L,M$ implies that the 1-forms
  $\chi, \omh, \xi$ must be smooth (otherwise their exterior
  derivatives (\ref{dchi}, \ref{domegahat}, \ref{dxi}) would not be).  Then, from \eqref{invariants},
  \eqref{solution}, \eqref{invg} and \eqref{smooth1} it follows that: (i) all
  components of the metric are smooth, (ii) a smooth inverse metric
  exists, and (iii) the Maxwell field is smooth.  Furthemore, equation
  \eqref{H} shows that 
  \be
  N^{-1} = K^2+HL\,,
  \ee 
  so \eqref{smooth1}
  is equivalent to $N>0$. Therefore we have established part
  1 of the Lemma. Part 2 immediately follows from
  (\ref{H},\,\ref{K},\,\ref{L},\,\ref{M}).
\end{proof}

\noindent {\bf Remarks} 
\begin{enumerate}
\item As we show in the next section, in the context of asymptotically
  flat, supersymmetric and biaxisymmetric spacetimes with a globally
  hyperbolic domain of outer communication, the invariant $N \geq 0$
  on and outside any black hole region and vanishes only in two
  instances: (i) at fixed points of the triholomorphic Killing field
  $\partial_\psi=0$, or (ii) on an event horizon. We will analyse (i)
  and (ii) later making more detailed use of the biaxial symmetry
  together with certain global constraints. 

\item We will require the spacetime to be stably causal,
  \begin{equation}
    g^{tt}<0 \label{causality}\,,
  \end{equation}
  which means $t$ is a time function and is equivalent to the absence
  of CTC. In particular, as can be seen from \eqref{invg}, stable
  causality and $N>0$ imply that $g_{\psi\psi}>0$.

\item A priori, the metric in the chart $(t, \psi, x^i)$ is only
defined in a region where $f \neq 0$ , i.e. where $V$ is
timelike. However, Lemma \ref{GHlemma} shows that in the region $N>0$ the metric
extends to a smooth solution even if $f=0$. In fact, \eqref{H} shows
that in the region $N>0$, the zero set of $f$ is precisely the locus
$H=0$. The zero set of $f$ is a smooth hypersurface in the spacetime if $\td f \neq 0$ everywhere on the zero set;\footnote{It is possible that the zero set of $f$ is not always a hypersurface; we will not analyse this possibility here. In any case, our analysis will not assume this. %and reveals that for asymptotically flat solutions with biaxial symmetry the zero set of $f$ in the DOC is always a hypersurface.
} then, from \eqref{H}, it follows that in the region $N>0$ we may use the harmonic function $H$ as
a coordinate in a neighbourhood of the zero set of $f$. In particular, we may introduce
coordinates $(H, y^A)$ on $\mathbb{R}^3$ valid near $H=0$ so that 
\begin{equation}
  \td x^i \td x^i = \varrho ^2\td H^2 + \tilde{g}_{AB} \td y^A \td y^B \,.
\end{equation}
The fact that $H$ is harmonic implies that
$\sqrt{\det{\tilde{g}}} = \varrho F(y)$ where $F$ is an arbitrary
function.  By a coordinate change $y^A\to y'^A(y)$ we may arrange
$\sqrt{\det{\tilde{g}}}|_{H=0} = F(y)$ so that
$\varrho = \sqrt{\det{\tilde{g}}}/\sqrt{\det{\tilde{g}}}|_{H=0} =
1+O(H)$.  The 1-form $\chi$ now satisfies
$\td \chi = \varrho^{-1} \tilde{\epsilon}$, where $\tilde{\epsilon}$ is
the volume form of $\tilde{g}$, so
$\chi = (\tilde{\chi}_A(y) +O(H)) \td y^A$ where
$\td \tilde{\chi} = \tilde{\epsilon}|_{H=0}$.  The spacetime metric
and gauge field induced on $H=0$ are
\begin{equation}
  \begin{aligned}
    \td s^2|_{H=0} &= - \frac{2}{K} (\td t+ {\hat{\omega}}_A \td
    y^A) ( \td \psi+ \tilde{\chi}_A \td y^A) + \frac{ 3L^2- 8 KM}{4
      K^2} (\td \psi+ \tilde{\chi}_A \td y^A)^2
    + K^2 \tilde{g}_{AB}\td y^A \td y^B  \,,\\
    A|_{H=0} &= \frac{\sqrt{3}}{2} \Big(\frac{L}{2K} ( \td \psi+ \tilde{\chi}_A \td y^A) -
    {\xi}_A \td y^A \Big)\,.
\end{aligned}
\end{equation}
Note that in such a region \eqref{smooth1} is satisfied, so $K$ must be
non-vanishing at $H=0$. Thus we see that $H=0$ is a  smooth timelike hypersurface. Therefore $H=0$ is an `evanescent' ergosurface, i.e. a timelike
hypersurface on which $f=0$~\cite{Niehoff:2016gbi}. In fact, it has been
shown~\cite{Niehoff:2016gbi} that any supersymmetric solution to
minimal supergravity is smooth at an evanescent ergosurface if and
only if the hyper-K\"ahler base is ambipolar (according to their
definition) and $\omega$ has a particular behaviour near the
ergosurface. 
\end{enumerate}

We will be interested in asymptotically flat solutions.   For
orientation, Minkowski space is given by  
\be
H=\frac{1}{r}\,, \qquad L=1\,, \qquad K=M=0\,, \qquad \chi = \cos
\theta \td \phi\,, \qquad \omh=\xi = 0\,,
\ee
where we have written the $\mathbb{R}^3$ base in polar coordinates
$(r, \theta, \phi)$. Upon the
coordinate change $r= \rho^2/4$ the metric is then
\begin{equation}
  \td s^2_{\text{Mink}}  = -\td t^2+ \td \rho^2 + \tfrac{1}{4} \rho^2 [ (\td
  \psi + \cos\theta \td \phi)^2  +\td \theta^2+\sin^2\theta \td \phi^2 ]\,,
\end{equation}
so the spatial $\mathbb{R}^4$  is in polar coordinates with the round $S^3$
written in terms of Euler angles $(\theta, \phi, \psi)$. These are
related to independently $2\pi$-periodic angles $\phi^\pm$ in
orthogonal planes via
\begin{equation}
\phi^\pm= \tfrac{1}{2} (\phi \mp \psi)\,, \qquad v_\pm
\equiv \partial_{\phi^\pm} = \partial_\phi
\mp \partial_\psi \,.  \label{vpm}
\end{equation}
Note $v_+=0$ on $\theta=0$ and $v_-=0$ on $\theta=\pi$ represent the
two axes which extend out to infinity. In terms of Euler angles the
periodicities are generated by the identifications $(\psi, \phi) \sim
(\psi+4\pi, \phi)$ and $(\psi, \phi) \sim (\psi+2\pi, \phi+2\pi)$. 

The asymptotic expansion of an asymptotically flat spacetime is
particularly simple for this class of metrics. Requiring that the
harmonic functions $H$, $K$, $L$, $M$ asymptotically approach those of
Minkowski spacetime (up to the freedom
(\ref{Kshift},\,\ref{LMshift})) implies they can be written as a standard
multipole expansion
\be
\begin{aligned}
  H &=  \frac{1}{r}+ \sum_{l \geq 1, m}  h_{l m}Y_{l m}(\theta, \phi)
      r^{-l-1}\,,  \label{Hasympt} \\ 
  K &=\frac{k_\infty}{r}+ \sum_{l \geq 1, m}  k_{l m}Y_{l m}(\theta,
      \phi) r^{-l-1}\,,  \\
  L &= 1+\frac{\ell_\infty}{r}+ \sum_{l \geq 1, m}  \ell_{l m}Y_{l
      m}(\theta, \phi) r^{-l-1}\,, \\ 
  M&= m+\frac{m_\infty}{r}+ \sum_{l\geq 1, m}  m_{l m}Y_{l m}(\theta,
     \phi) r^{-l-1}\,,
\end{aligned}
\ee
where $k_\infty, \ell_\infty, m_\infty, m$ are constants. The constant $k_\infty$ is pure gauge and can be set to any value using \eqref{Kshift}. We have included a constant $m$ in $M$ in order not to fix the corresponding gauge freedom \eqref{LMshift}. The above form for the harmonic functions then determines the asymptotic expansion of the spacetime.  In particular, this implies $f = 1+ O(r^{-1})$ so  the supersymmetric Killing vector $V$ is timelike near infinity, i.e. $V$ is the stationary Killing field.   Furthermore, $\omega_\psi = m+\tfrac{3}{2} k_\infty + O(r^{-1})$ so we set 
\begin{equation}
m = - \tfrac{3}{2} k_\infty \,,  \label{masympt}
\end{equation}
which is indeed invariant under the gauge transformations \eqref{Kshift} and \eqref{LMshift}.
When integrating for $\chi$ and $\omh$ we will also set the integration constants
so that $\chi = \cos \theta \td \phi+ O(r^{-1})$ and $\omh= O(r^{-1})$. Without loss of generality, we may always make such choices for asymptotically flat solutions. It is now clear that the smoothness condition (\ref{smooth1}) and causality condition (\ref{causality}) are satisfied in the asymptotic region.

\subsection{Supersymmetric and biaxisymmetric spacetimes}

We now impose that  a (timelike) supersymmetric background $(M,g,F)$ also has biaxial symmetry. As we will explain, this implies the solution must be of the type studied in the previous section (i.e. Gibbons--Hawking type).  

In
particular, we assume: (i) there is a $U(1)^2$-isometry, generated by
Killing fields $m_i$, $i=1,2$, whose orbits are $2\pi$ periodic; (ii)
$[ V, m_i ]=0$ where $V$ is the supersymmetric Killing field; (iii)
$\mathcal{L}_{m_i} F=0$. Clearly, the biaxial Killing fields are
defined up to $m_i \to A_{ij} m_j$ where $A_{ij}$ is an $SL(2,
\mathbb{Z})$ matrix. We will sometimes denote the Killing fields
collectively by $K_A$, where $A=0,1,2$ and $K_0=V$ and $K_i=m_i$. 

The above conditions mean the spacetime is stationary and biaxisymmetric, where the supersymmetric Killing field $V$ is the stationary Killing field. In the context of (electro-)vacuum gravity these correspond to the well studied (generalised) Weyl solutions. Therefore we are dealing with supersymmetric Weyl solutions. The additional assumption of supersymmetry places extra local and global constraints on the solution which we will now explore.

Firstly, supersymmetry together with biaxial symmetry places strong constraints on the local form of the solution, as follows.
\begin{lemma}
  Consider an asymptotically flat, supersymmetric and biaxisymmetric
  solution to minimal supergravity with supersymmetric Killing field
  $V=\partial_t$. Then the hyper-K\"ahler base must be a
  Gibbons--Hawking metric \eqref{GH} whose triholomorphic Killing field
  $\partial_\psi$ leaves the full solution invariant. Furthermore, the
  harmonic functions $H,K,L,M$ on $\mathbb{R}^3$ are axisymmetric and
  the 1-forms can be written as\footnote{For notational simplicity we
    denote both the $1$-forms and their $\phi$-components by $\chi$,
    $\omh$, $\xi$. Distinction between these will be apparent from
    context, or clarified if necessary.}
\begin{equation}
  \chi = \chiph(\rho, z) \td \phi\,, \qquad \omh = \omhph(\rho,z) \td
  \phi\,, \qquad \xi  = \xi(\rho,z) \td \phi\,,    \label{1forms}
\end{equation} 
where $(\rho, z, \phi)$ are cylindrical coordinates on $\mathbb{R}^3$.
In particular, in the coordinates $y^A= (t, \psi, \phi)$ and  $z^a=(\rho,z)$, the spacetime metric \eqref{g} then takes the block diagonal form
\begin{align}
\td s^2 = G_{AB}(z^a)\td y^A \td y^B+  q_{ab}(z^a) \td z^a \td z^b\, ,   \label{block}
\end{align}
where
\be
G_{AB} \td y^A \td y^B  = - f^2 ( \td t + \omega_\psi (\td\psi+\chi \td\phi) + \omhph \td\phi)^2 + f^{-1} H^{-1} ( \td\psi + \chi \td\phi)^2 + \frac{H \rho^2}{f} \td\phi^2 
\ee
 is the inner product on the space spanned by the Killing fields $\{\partial_t, \partial_\psi, \partial_\phi \}$, and
\be
q_{ab} \td z^a \td z^b =  \frac{H}{f} ( \td\rho^2 + \td z^2) 
 \label{orbit}
 \ee
is a metric on surfaces orthogonal to the space of Killing fields.
The determinant of $G_{AB}$ is
\begin{equation}
\det G_{AB}  = -\rho^2  \, ,   \label{rho}
\end{equation}
so $(\rho,z)$ are in fact standard Weyl coordinates. 
\end{lemma}

\begin{proof} First we show that for a supersymmetric solution \eqref{g} biaxial symmetry requires  the hyper-K\"ahler base $h$ and the data $f$, $\omega$ defined on it, to be $U(1)^2$-invariant.  By assumption $\mathcal{L}_{m_i} V=0$ and therefore $f=\sqrt{- V_\mu V^\mu}$ implies $\mathcal{L}_{m_i} f=0$. Writing the hyper-K\"ahler metric as
\begin{equation}
h_{\mu\nu} = f g_{\mu\nu} + f^{-1} V_\mu V_\nu
\end{equation}
we then immediately deduce $\mathcal{L}_{m_i} h_{\mu\nu}=0$. Finally, by a shift $t \to t +\lambda$, where $\lambda$ is a function on $B$, we may choose the time function so $\mathcal{L}_{m_i}t=0$, which also implies $\mathcal{L}_{m_i}\omega=0$ (by taking the Lie derivative of $V_\mu \td x^\mu = - f^2 (\td t+\omega)$).  

We may now apply the following result~\cite{Gibbons:1987sp}:  {\it A hyper-K\"ahler metric with a local $U(1)^2$-isometry is a Gibbons--Hawking metric and the triholomorphic Killing field is a linear combination of the $U(1)^2$ Killing fields.} 
Therefore, biaxial symmetry of the full solution implies that the base must be a GH metric (\ref{GH}) and also that the full solution is invariant under the triholomorphic Killing field. Hence, the local form of the solution is determined by four harmonic functions $H,K,L,M$ on $\mathbb{R}^3$ as discussed in section \ref{sec:GH}.

The GH metrics in general only possess a $U(1)$-isometry generated by the triholomorphic Killing field. This allows us to write the harmonic function $H$ and the $\mathbb{R}^3$ base as the invariants $H^{-1} = h (\partial_\psi, \partial_\psi)$ and $\delta= H^{-1}  h - h(\partial_\psi, \cdot)^2$ respectively. Hence, the $U(1)^2$-symmetry implies that $H$ and $\delta$ are invariant under a linear combination of the $U(1)^2$ Killing fields, say $\eta$, which is linearly independent to $\partial_\psi$.   By shifting $\psi \to \psi + \lambda$, where $\lambda$ is a function on $\mathbb{R}^3$, we may set $\mathcal{L}_\eta \psi=0$. Then $\eta$ is a vector field on $\mathbb{R}^3$ which leaves $H$ and $\delta$ invariant, i.e. $\eta$ is a Killing vector of $\mathbb{R}^3$. 

Hence we have shown that the $U(1)^2$-symmetry assumption implies the harmonic function $H$ is  invariant under a 1-parameter group of isometries of $\mathbb{R}^3$.  If this 1-parameter subgroup is closed then it must be a subgroup $U(1) \in SO(3)$, i.e.~the harmonic function must be axisymmetric.  What if $H$ is invariant under a non-compact  1-parameter subgroup of the Euclidean group? Such subgroups correspond to a translation, or a combination of a translation with a rotation (corkscrew). In either case, the orbit curves of such subgroups are unbounded in $\mathbb{R}^3$. Since $H$ is invariant along such an orbit curve, and by asymptotic flatness $H \to 0$ at infinity, we deduce that $H \equiv 0$ everywhere, a contradiction.  This argument shows that the only 1-parameter subgroup of the Euclidean group which may leave $H$ invariant is the axial symmetry.  Thus we write the $\mathbb{R}^3$ in cylindrical coordinates 
\begin{equation}
\td x^i \td x^i = \td \rho^2 + \rho^2 d\phi^2 + \td z^2 \, ,
\end{equation}
where the axial Killing field is $\eta=\partial_\phi$. Then $H$ is invariant under  $\partial_\phi$ and hence is only a function of $(\rho,z)$. We will now show that the other harmonic functions $K,L,M$ must also be axisymmetric, i.e. invariant under the axial Killing field $\partial_\phi$.  

First, we recall some well known properties of Maxwell fields invariant under three commuting Killing fields $K_A$, see e.g.~\cite{Kunduri:2013vka}.  The Bianchi identity implies that the functions $\iota_{K_A} \iota_{K_B} F$ are constant in the spacetime.  Furthermore, asymptotic flatness implies that a different linear combination of $K_i=m_i$, for $i=1,2$, vanish on the two axes which intersect the $S^3$ at infinity, namely $v_+$ and $v_-$ in \eqref{vpm}. Therefore all these constants must in fact vanish, so $\iota_{K_A} \iota_{K_B} F=0$.  In particular, since the axial Killing field $\partial_\phi$ and the triholomorphic Killing field $\partial_\psi$ must be linear combinations of the $K_i=m_i$ for $i=1,2$, we must have $\iota_{\partial_\phi} \iota_{\partial_\psi} F=0$ so by \eqref{PhiPsi} the magnetic potential is axisymmetric $\mathcal{L}_{\partial_\phi} \Psi=0$. Hence \eqref{potentials} implies that $K H^{-1}$, and
thus $K$, is also axisymmetric. Axisymmetry of $L$ and $M$ then follows from invariance of $f$ and $\omega_\psi$ under $\partial_\phi$, together with \eqref{Ldef} and \eqref{Mdef}.

Furthermore, axisymmetry implies that the 1-forms $\chi$, $\omh$, $\xi$, 
are all gauge equivalent to (\ref{1forms}). To see this, first note that axisymmetry of the harmonic function $H$ and invariance of the 1-form $h(\partial_
\psi, \cdot) = H^{-1} (\td \psi+\chi)$ under $\partial_\phi$, together with our gauge choice $\mathcal{L}_{\partial_\phi} \psi=0$, implies that  $\mathcal{L}_{\partial_\phi} \chi = 0$. Also, $0=\iota_{\partial_\phi} \td H= \iota_{\partial_\phi} \star_3 \td \chi \sim \rho^2 \star_3 ( \td \phi \wedge \td \chi)$, which implies $\chi = \chiph(\rho, z) \td \phi+ \td (\lambda'(\rho, z))$ for some function $\lambda'(\rho, z)$. By a shift in $\psi \to \psi - \lambda'$ (which does not affect $\mathcal{L}_{\partial_{\phi}} \chi=0$) we may eliminate $\lambda'$. A similar argument works for the 1-forms $\omh$ and $\xi$ by shifting $t$ and the gauge field $A$ respectively, establishing the claim. 

Putting everything together we find the spacetime metric can be written as claimed.  The local form of the metric \eqref{block} shows the distribution orthogonal to $\text{span}\{ \partial_t, \partial_\psi, \partial_\phi \}$ is integrable so that at every point there exist surfaces orthogonal to the Killing fields with metric (\ref{orbit}).\footnote{This is equivalent to the Frobenius integrability condition $K_0 \wedge K_1 \wedge K_2 \wedge \td K_A=0$, which is in fact guaranteed for any
solution to the Einstein--Maxwell equation with $D-2$ commuting Killing
fields, one of which has at least one fixed point (which must be the case here due to asymptotic flatness), see e.g.~\cite{Tomizawa:2009ua}. Here it arises as a
consequence of supersymmetry which for timelike solutions implies the
Einstein--Maxwell equations~\cite{Gauntlett:2002nw}.}  
\end{proof}

We now turn to our global assumptions. We assume the spacetime $(M,g)$ is asymptotically flat  and the domain of outer communication $\llangle M \rrangle$ is globally hyperbolic so $\llangle M \rrangle \cong \mathbb{R}\times \Sigma$. 
Topological censorship implies that $\llangle M \rrangle$  is simply connected \cite{Friedman:1993ty}.  We will denote the event horizon  by $\mathcal{H}$, although we allow for the possibility of no black hole region. We will assume the stationary Killing vector $V$ is complete so the spacetime has an isometry group
$G=\mathbb{R}\times U(1)^2$, where $\mathbb{R}$ is tangent to the orbits of $V$.
The axes correspond to the set of fixed points of the biaxial symmetry
\begin{equation}
\mathcal{A} =\{ p \in M \; | \; \det \gamma_{ij}(p)=0 \}
\end{equation}
where $\gamma_{ij} = m_i \cdot m_j$ and $i,j=1,2$. Under these conditions, it has been shown that the orbit space 
\begin{equation}
\hat{M} = \llangle M \rrangle  / G \cong \Sigma /U(1)^2
\end{equation}
is a simply connected 2-dimensional manifold with boundaries and corners~\cite{Hollands:2007aj, Hollands:2008fm, Hollands:2012xy}. 
The axes corresponds to boundary segments $I \subset \partial \hat{M}$
where $\gamma_{ij}$ is rank-1 and to corners of $\hat{M}$ where
$\gamma_{ij}$ is rank-0.  Below we will show that an event horizon,
which must be degenerate, corresponds to a point on $\partial \hat{M}$
(in fact an asymptotic end, as is generic for extremal horizons, see e.g.~\cite{Figueras:2009ci}).

Now, we may identify the surfaces orthogonal to the Killing fields with the the orbit space $\hat{M}$. We deduce that the orbit space inherits the metric $q$ (\ref{orbit}), so we will refer to this as the orbit space metric. Under the above global assumptions, it has been shown that $\det G_{AB}<0$ everywhere on
$\llangle M \rrangle \setminus \mathcal{A}$ and $\det G_{AB}=0$ on
$\mathcal{H}  \cup
\mathcal{A}$~\cite{Chrusciel:2008rh}.
 From (\ref{rho}) we immediately deduce that  $\rho>0$ everywhere on $\llangle M \rrangle \setminus \mathcal{A}$ and $\rho=0$ on
$\mathcal{H} \cup \mathcal{A}$. Thus
the interior of the orbit space corresponds to $\rho>0$ and its
boundary to $\rho=0$. Therefore, $(\rho, z)$ can be used as global coordinates on the interior of $\hat{M}$, so we may
identify the interior of the orbit space with the upper-half plane 
\begin{equation}
\hat{M}  = \{ (\rho, z) \; | \; \rho > 0 \}
\end{equation}
and the boundary $\partial \hat{M}$ and corners with the $z$-axis ($\rho=0$). In the orbit space the axes divide into boundary segments $I = (z_1, z_{2})$, called axis rods (or intervals), and corners which arise as certain endpoints $z=z_i$ of the axis rods. Below we will show that in the orbit space an event horizon is a point on the $z$-axis.

Since $\det G_{AB}<0$ on $\llangle M \rrangle \setminus \mathcal{A}$, the orbit space metric \eqref{orbit} must be Riemannian on $\hat{M}$ and therefore its conformal factor must be a smooth and positive function for $\rho>0$, i.e.~
\begin{equation}
\frac{H}{f} >0  \label{orbitsmooth}\,.
\end{equation}
In fact, from \eqref{H}, we see that \eqref{orbitsmooth} is equivalent to the invariant $N>0$.  Thus the above shows that $N>0$ on $\llangle M \rrangle \setminus \mathcal{A}$. Therefore, Lemma \ref{GHlemma} may be applied to learn that the harmonic functions $H,K,L,M$ are smooth on  $\llangle M \rrangle \setminus \mathcal{A}$. 

We require that the spacetime metric is smooth at the axes. Consider an axis rod $I$ and let $v=v_i m_i$, where $(v_1, v_2)\in \mathbb{Z}^2$, denote the Killing field which vanishes on $I$. Smoothness requires that for each such axis rod $I$ 
\begin{equation}
 \frac{H}{f}  = \lim_{\rho \to 0} \frac{\gamma_{ij} v_i v_j}{\rho^2}  \label{smoothaxis}
\end{equation}
is a smooth positive function for all $z \in I$\ \footnote{This condition has been previously derived for vacuum Weyl solutions \cite{Harmark:2004rm, Hollands:2008fm, Hollands:2012xy}.}.
 To see this, consider the spacetime metric for fixed $z \in I$,
\begin{equation}
\td s^2= \frac{H}{f} \td \rho^2 + G_{tt} \td t^2+ 2 G_{t i} \td t \td \phi^i +G_{ij} \td \phi^i \td \phi^j\,.
\end{equation}
By an $SL(2,\mathbb{Z})$ transformation we may assume that $v=m_1=\partial_{\phi^1}=0$ on $I$, so $G_{A1}=0$ on $I$ for all $A$. The metric is smooth on $I$ provided $(\td |m_1 |)^2 \to 1$ as $\rho \to 0$. In terms of the proper distance $s= \int_0^\rho \sqrt{g_{\rho\rho}} \td \rho$ this implies $G_{11}= s^2+O(s^4)$ and $G_{1 t}=G_{1 i}= O(s^2)$. Therefore
\begin{equation}
 -\rho^2= \det G =  \left| \begin{array}{cc} G_{tt} & G_{t2} \\ G_{2t} &  G_{22} \end{array} \right|  s^2 + O(s^4)  \,.\label{nearI}
\end{equation}
Now, it has been shown that $\text{span} \{ K_0, K_1, K_2 \}$ is timelike everywhere in $\llangle M \rrangle$~\cite{Chrusciel:2008rh}. Therefore, $\text{span} \{ V, m_2 \}$ must be timelike on $I$ and hence the determinant on the righthand side of \eqref{nearI} is strictly negative on $I$. Thus, 
\begin{equation}
\frac{f}{H}  = \left(\frac{\td \rho}{\td s} \right)^2 = -  \left| \begin{array}{cc} G_{tt} & G_{t2} \\ G_{2t} &  G_{22} \end{array} \right|  + O(s^2)
\end{equation}
 is smooth and positive on $I$. The condition \eqref{smoothaxis} now easily follows. From (\ref{H}) we deduce that the invariant $N>0$  on all axis rods. 
 
 Thus we arrive at the following crucial result.

\begin{lemma}
Let $(M,g,F)$ be an asymptotically flat, supersymmetric and biaxisymmetric solution to minimal supergravity with a globally hyperbolic $\llangle M \rrangle$. 
\begin{enumerate}
\item The fixed points of the triholomorphic Killing field of the Gibbons--Hawking base correspond to precisely the corners of the orbit space $\hat{M}$.
\item The harmonic functions $H,K,L,M$ are smooth and obey \eqref{smooth1} everywhere in $\llangle M \rrangle$ except possibly at points corresponding to the corners of the orbit space $\hat{M}$.
\item At every corner of the orbit space $f\neq0$ and $H$ has an isolated singularity.
\end{enumerate}
\label{lemma:orbitspace}
\end{lemma}

\begin{proof}
The fixed points of $\partial_\psi$ are on $\mathcal{A}$ and therefore must either occur on an axis rod or on a corner of $\hat{M}$.  But $\partial_\psi=0$ implies that the invariant $N=0$ (recall (\ref{I})), which as we have shown above cannot occur on an axis rod. Therefore, $\partial_\psi=0$ can only occur at a corner of $\hat{M}$. Conversely, by definition, at any corner $m_1=m_2=0$ and hence $\partial_\psi=0$. Thus we have proved part 1.

Next, above we observed that the invariant $N>0$ on $\llangle M \rrangle \setminus \mathcal{A}$ and also on the parts of $\mathcal{A}$ corresponding to the axis rods. Thus by application of Lemma \ref{GHlemma} we deduce that $H,K,L,M$ are smooth and obey \eqref{smooth1} on $\llangle M \rrangle \setminus \mathcal{A}$ and also on the parts of $\mathcal{A}$ corresponding to the axis rods.  Thus the only potential singularities of the harmonic functions in  $\llangle M \rrangle$ are points corresponding to the corners of the orbit space, establishing 2.

Finally, we again use the fact that $\text{span} \{ K_0, K_1, K_2 \}$ is timelike everywhere in $\llangle M \rrangle$~\cite{Chrusciel:2008rh}. In particular, this implies that $K_0=V$ must be timelike at any corner of $\hat{M}$. Therefore, $f \neq 0$ at any corner. Also, as observed above $N=0$  at the corners. Therefore, the expression for $H$ in terms of these invariants \eqref{H} shows that $H$ {\it must} be singular at the corners of $\hat{M}$. Since we have already established that $H$ is smooth on the orbit space away from the corners and horizon, we deduce the singularities of $H$ are isolated. 
\end{proof}

Therefore the analysis reduces to studying the behaviour of the harmonic functions at the event horizon and the corners the of orbit space.

\section{Near-horizon geometry}
\label{sec:hor}

We will now examine the geometry near a horizon in detail.  The event horizon $\mathcal{H}$ of a black hole must be invariant under the isometries of the spacetime. Hence any Killing field must be tangent to $\mathcal{H}$, which  implies it must be null or spacelike on $\mathcal{H}$. In particular, the supersymmetric Killing field $V$ must be tangent and hence null on $\mathcal{H}$ (since it is never spacelike). Hence $V$ is also normal to the horizon and $\mathcal{H}$ is a Killing horizon with respect to $V$. Furthermore, $\td (V \cdot V) =0$ on $\mathcal{H}$, so the horizon must be degenerate (i.e.~extreme). We will refer to such horizons as supersymmetric horizons. Since for asymptotically flat solutions $V$ is also the stationary Killing field we deduce that any supersymmetric black hole is nonrotating.

In the neighbourhood of a supersymmetric horizon we may introduce Gaussian null coordinates $(v,\lambda, x^a)$~\cite{Kunduri:2013ana}
\begin{equation}
\td s^2= -\lambda^2 \Delta(\lambda, x)^2 \td v^2+ 2 \td v \td \lambda+ 2 \lambda h_a(\lambda, x) \td v \td x^a+ \gamma_{ab}(\lambda, x) \td x^a \td x^b\,,
\end{equation}
where $V= \partial_v$, $\lambda=0$ is the horizon and $f= \lambda
\Delta(\lambda, x)$. The supersymmetric near-horizon geometries of minimal supergravity have
been classified~\cite{Reall:2002bh}. Assuming cross-sections of the
horizon are compact, it can be shown that $\Delta|_{\lambda=0} =
\Delta_0$ is a constant on the horizon. If $\Delta_0 \neq 0$ the
horizon is locally $S^3$, and if $\Delta_0=0$ it is $S^1\times S^2$
(we do not consider the $T^3$ case as this is not an allowed topology
for black holes)~\footnote{The analysis of~\cite{Reall:2002bh} assumes $f>0$ (i.e. $\Delta>0$) for small $r>0$. However, this is restrictive and we should only assume $f \neq 0$ (i.e. $\Delta \neq 0$ for small $r>0$). In fact, the analysis of ~\cite{Reall:2002bh} remains valid under this weaker assumption since only $\Delta^2$ appears in the near-horizon geometry.}. In fact, an output of this analysis is that the near-horizon geometry must have biaxial symmetry.  We will use this classification below after a general analysis of the orbit space.

\subsection{Orbit space metric}

Let us now consider the orbit space metric near an extreme horizon. Our analysis in this section will be general and only assume the existence of biaxial symmetry.

The biaxial Killing fields must be tangent to the horizon and
hence on a cross-section we may introduce coordinates adapted to this
symmetry. We thus introduce Gaussian null coordinates
$(v,\lambda, \tilde{\theta}, \tilde{\phi}^i)$ where
$(\tilde{\theta}, \tilde{\phi}^i)$, $i=1,2$, are such that the three
commuting Killing fields are
$V=\partial_v, m_i=\partial_{\tilde{\phi}^i}$.  Then we write
\begin{multline}\label{gnh}
\td s^2 = -\lambda^2 F \td v^2 + 2 \td v ( \td \lambda+ \lambda
\tilde{h}_{\tilde{\theta}} \td \tilde{\theta})+ \gamma_{\tilde{\theta}
  \tilde{\theta}} \td \tilde{\theta}^2  
 + 2\gamma_{\tilde{\theta} i} \td \tilde{\theta} ( \td \tilde{\phi}^i
 + \lambda h^i \td v)  \\
 + \gamma_{ij} ( \td \tilde{\phi}^i + \lambda
 h^i \td v) ( \td \tilde{\phi}^j + \lambda h^j \td v)\,,
\end{multline}
where $F = \Delta^2+h^i h_i$, $\tilde{h}_{\tilde{\theta}}  = h_{\tilde{\theta}}- \gamma_{\tilde{\theta} i} h^i$ and $h^i = \gamma^{ij} h_j$ where $\gamma^{ij}$ is the inverse of the 2d matrix $\gamma_{ij}$.  We will now extract the orbit space metric, following~\cite{Chrusciel:2010gq}.

The inner product on the space of Killing fields is
\begin{equation}
G = -{\lambda}^2 {F} \td v^2 + \gamma_{ij} ( \td \tilde{\phi}^i + {\lambda} {h}^i \td v) ( \td \tilde{\phi}^j + {\lambda} {h}^j \td v)  
\end{equation}
with inverse
\begin{equation}
G^{AB} = \left( \begin{array}{cc} - \frac{1}{{\lambda}^2 {F}} & \frac{{h}^i}{{\lambda} {F}} \\  \frac{{h}^j}{{\lambda} {F}} & \gamma^{ij} - \frac{{h}^i {h}^j}{{F}} \end{array} \right)\,.
\end{equation}
The orbit space metric may be defined by
\begin{align}
q_{\mu\nu}  = g_{\mu\nu} - G^{AB} g_{A \mu} g_{B \nu} 
\end{align}
so $q_{A \nu}=0$. Computing we find
\begin{align}
q &= \frac{1}{F \lambda^2}( \td {\lambda} + \lambda \tilde{h}_{\tilde{\theta}} \td \tilde{\theta})^2+ (\gamma_{\tilde{\theta}\tilde{\theta}} - \gamma^{ij} \gamma_{i \tilde{\theta}} \gamma_{j\tilde{\theta}}) \td \tilde{\theta}^2 \,.  \label{orbitgnc}
\end{align}
An alternate way to derive this orbit space metric is as follows.

The spacetime has three commuting Killing fields. Hence the distribution orthogonal to these Killing fields is integrable and there exist local coordinates in which the metric takes block diagonal form as  \eqref{orbit} (see e.g.~\cite{Tomizawa:2009ua}). The general coordinate change which takes us from Gaussian null coordinates to block diagonal coordinates is
\begin{equation}
t = v+ A(\lambda, \tilde{\theta})\,, \qquad \phi^i = \tilde{\phi}^i+B^i(\lambda, \tilde{\theta})\,,
\end{equation}
where $V= \partial_t$, $m_i= \partial_{\phi^i}$ and
\begin{equation}
\partial_\lambda A  = - \frac{1}{\lambda^2 F}, \qquad \partial_{\tilde{\theta}} A = - \frac{\tilde{h}_{\tilde{\theta}}}{\lambda F}, \qquad \partial_\lambda B^i = \frac{h^i}{\lambda F}, \qquad \partial_{\tilde{\theta}} B^i = \gamma^{ij} \gamma_{\tilde{\theta} j}+ \frac{h^i \tilde{h}_{\tilde{\theta}}}{F}\,.
\end{equation}
A calculation then shows that the spacetime metric takes the form \eqref{block} with the matrix of Killing fields
\begin{equation}
G = -{\lambda}^2 {F} \td t^2 + \gamma_{ij} ( \td {\phi}^i + {\lambda} {h}^i \td t) ( \td {\phi}^j + {\lambda} {h}^j \td t) 
\end{equation}
and the orbit space metric given by \eqref{orbitgnc}.

We will now examine the orbit space metric near the horizon. We will
assume $h_{\tilde{\theta}} = O(\lambda)$ and $\gamma_{i
  \tilde{\theta}} = O(\lambda)$ which are conditions satisfied by the
near-horizon geometries in question.  We find that  near the horizon
\eqref{orbitgnc} takes the form
\begin{equation}
q = \left( \frac{1}{F \lambda^2} +O(\lambda^{-1}) \right) \td \lambda^2 +O(1) \td \lambda \td \tilde{\theta} + (\gamma_{\tilde{\theta}\tilde{\theta}}  +O(\lambda)) \td \tilde{\theta}\,,  \label{orbitnh}
\end{equation}
and the determinant of the Killing metric is
\begin{align}
\rho &= \sqrt{- \det G} = {\lambda} \sqrt{ {F} \det \gamma_{ij}}  \nonumber\\ &=   \sqrt{ {F} \det \gamma_{ij}}|_{\lambda=0} \; \lambda + O(\lambda^2)\,.   \label{rhonh}
\end{align}
The function $\rho$ is harmonic in the orbit space metric \eqref{orbit}. The harmonic conjugate $z$ is given by $\td z = \star_2 \td \rho$.  A computation gives
\begin{align}
\partial_\lambda z &= \frac{1}{\lambda \sqrt{F q_{\tilde{\theta}\tilde{\theta}}}} \left(\partial_{\tilde{\theta}} \rho -\lambda \tilde{h}_{\tilde{\theta}} \partial_\lambda \rho \right)  = \left. \frac{\partial_{\tilde{\theta}} \sqrt{ {F} \det \gamma_{ij}} }{\sqrt{F \gamma_{\tilde{\theta}\tilde{\theta}}}} \right|_{\lambda=0} + O(\lambda) \,, \label{znh}\\
\partial_{\tilde{\theta}} z &=  \frac{1}{\sqrt{F q_{\tilde{\theta}\tilde{\theta}}}} \left[ - \lambda F q_{\tilde{\theta}\tilde{\theta}}  \partial_\lambda \rho + \tilde{h}_{\tilde{\theta}} \left(\partial_{\tilde{\theta}} \rho -\lambda \tilde{h}_{\tilde{\theta}} \partial_\lambda \rho \right)\right]  = - \left. F \sqrt{\gamma_{\tilde{\theta}\tilde{\theta}} \det \gamma_{ij}} \right|_{\lambda=0} \; \lambda + O(\lambda^2) \,. \nonumber
\end{align}
Observe that by integrating for $z$ we see that it will take the form $z=z_0 + O(\lambda)$, where $z_0$ is a constant which can be set to zero. Hence, a degenerate horizon corresponds to a single point on the boundary of the orbit space. In fact, in the orbit space metric, it is easy to see that any point is an infinite proper distance to the horizon, so that a degenerate horizon corresponds to an asymptotic end (see~\cite{Figueras:2009ci} for discussion of this in the vacuum case).

To proceed further we need the specific near-horizon geometries. We will turn to this next.

\subsection{Locally $S^3$ horizon}\label{sec:S3horizon}

For a locally $S^3$ horizon, the horizon data is given by~\cite{Reall:2002bh}
\begin{equation}
  \begin{gathered}
    \Delta_0^2 = \frac{4}{\mu}  \left( 1- \frac{j^2}{\mu^3} \right) \,,\qquad 
    h|_{\lambda=0} = -j \mu^{-3/2}  \left( 1- \frac{j^2}{\mu^3}
    \right)^{1/2}( \td \tilde{\psi}+ \cos \tilde{\theta} \td
    \tilde{\phi}) \,, \\
    \gamma|_{\lambda=0} = \frac{\mu}{4} \left[  \left( 1-
        \frac{j^2}{\mu^3} \right) ( \td \tilde{\psi}+  \cos
      \tilde{\theta} \td \tilde{\phi})^2+ \td \tilde{\theta}^2+
      \sin^2\tilde{\theta}\td \tilde{\phi}^2 \right]\,,
  \end{gathered}
\end{equation}
where the constants $j^2<\mu^3$.   We deduce that near the horizon the
metric  is given by \eqref{gnh} where
\begin{equation}
  \begin{gathered}
    {F} = \frac{4}{\mu} +O(\lambda),  \qquad h^i =- \frac{4 j
    }{\mu^{5/2} \left( 1 - \frac{j^2}{\mu^3} \right)^{1/2} }
    \delta^i_{\tilde{\psi}} + O(\lambda)\,, \\
    h_{\tilde{\theta}} = O(\lambda), \qquad \gamma_{\tilde{\theta}
      \tilde{\theta}} = \frac{\mu}{4}+O(\lambda) ,  \qquad
    \gamma_{\tilde{\theta} i}=O(\lambda) \,,\\
    \gamma_{ij} \td \tilde{\phi}^i \td \tilde{\phi}^j =
    \frac{\mu}{4} \left[  \left( 1- \frac{j^2}{\mu^3} \right) ( \td
      \tilde{\psi}+  \cos \tilde{\theta} \td \tilde{\phi})^2+ \sin^2
      \tilde{\theta} \td \tilde{\phi}^2 \right] +O(\lambda)\,.
  \label{nh1data}
  \end{gathered}
\end{equation}
The above horizon geometry is locally that of a squashed $S^3$ with a
$U(1)^2$-isometry generated by the Killing fields
$(\partial_{\tilde{\psi}}, \partial_{\tilde{\phi}})$.  The topology of
the horizon is determined by the periodicity lattice of the biaxial
Killing fields. For now our analysis will be local, but we will
consider global constraints at the end of this section. In any case,
locally, biaxial symmetry implies that $(\partial_{\tilde{\psi}}, \partial_{\tilde{\phi}})$ must be related to the GH space biaxial Killing fields $(\partial_\psi, \partial_\phi)$ by a constant linear transformation. Hence, $\det \gamma_{ij}|_{i,j= \tilde{\psi},
  \tilde{\phi}} = c^2 \det \gamma_{ij}|_{i,j= {\psi}, {\phi}}$ for some constant $c>0$.
 We need to take account of this Jacobian when comparing the determinants
of the matrix of Killing fields to determine the Weyl coordinates
$(\rho, z)$.  

Using the near horizon data \eqref{nh1data}, we find that
\eqref{rhonh} and \eqref{znh} imply
\begin{equation}\label{rhoznh1}
  \rho= \sqrt{\frac{\mu}{4 c^2} \left( 1-\frac{j^2}{\mu^3}\right) }
  \lambda \sin \tilde{\theta} + O(\lambda^2)\,, \qquad 
  z = \sqrt{\frac{\mu}{4 c^2} \left( 1-\frac{j^2}{\mu^3}\right) } \lambda
  \cos \tilde{\theta}+ O(\lambda^2) \,.
\end{equation}
Thus the $\mathbb{R}^3$ polar coordinates are
\begin{equation}
  r =  \sqrt{\frac{\mu}{4 c^2} \left( 1-\frac{j^2}{\mu^3}\right) }
  \lambda   + O(\lambda^2)\,, \qquad
  \cos \theta = \cos \tilde{\theta} + O(\lambda)\,.
\end{equation}
We deduce that the horizon is a single point in the orbit space
metric, as anticipated above.  The orbit space metric in the $(\rho,z)$ coordinates must
take the form \eqref{orbit}, so using the coordinate change \eqref{rhoznh1} we find
\begin{equation}
  q = \frac{H}{f} \left[ \frac{\mu}{4 c^2} \left( 1- \frac{j^2}{\mu^3}
    \right)  ( \td \lambda^2 + \lambda^2 \td \tilde{\theta}^2)+
    O(\lambda) \td \lambda^2 + O(\lambda^2) \td \lambda \td
    \tilde{\theta}+ O(\lambda^3) \td \tilde{\theta}^2 \right]
\end{equation}
and comparing to \eqref{orbitnh} implies
\begin{equation}
  \frac{H}{f} = \frac{c^2}{1- \frac{j^2}{\mu^3}}  \left(
    \frac{1}{\lambda^2} + O(\lambda^{-1}) \right)\,.
\end{equation}
Therefore, using $f = \lambda \Delta_0+ O(\lambda^2)$ we deduce
\begin{equation}\label{NHH}
 H=   \frac{4 c^2}{\mu \Delta_0 \lambda} + O(1) = \frac{h_0}{r} + O(1) \,,
\end{equation}
where we have defined the constant $h_0 = \sgn (\Delta_0)
c$. Observe that the first term in $H$ is harmonic and hence the $O(1)$ term
is also a harmonic function.  Therefore, we deduce that a horizon
corresponds to an isolated singularity of the harmonic function $H$;
specifically the horizon is a pole of order one.

We now turn to the remaining harmonic functions. Firstly, observe that
we can write
$\partial_\psi = a \partial_{\tilde{\psi}}+ b \partial_{\tilde{\phi}}$
for some constants $a,b$. Hence the invariant
$g_{t\psi} = V \cdot \partial_\psi = \lambda h \cdot \partial_\psi=
\lambda (a h_{\tilde{\psi}}+ b h_{\tilde{\phi}})$. Then, the
near-horizon expansion of the invariants $g_{t\psi}$ and $f$, together
with \eqref{K} and smoothness of $\Psi$, imply that $K H^{-1}$ is
smooth at $\lambda=0$.  Therefore $K = O(r^{-1})$ and since it is
harmonic it must have a pole of order one, so
\begin{equation}\label{NHK}
K = \frac{k_0}{r} +O(1)\,,
\end{equation}
where $k_0$ is a constant and the $O(1)$ term is harmonic.  Due to the shift freedom in $K$ \eqref{Kshift} the constant $k_0$ can be set to any value.  

Next, using the expansion of the invariants $f$ and $- f^2 \omega_\psi = V \cdot \partial_\psi$ together with \eqref{Ldef} and \eqref{Mdef} implies
\begin{equation}\label{NHLM}
L =  \frac{\ell_0}{r}  +O(1), \qquad
M=\frac{1}{r} \left( \frac{j(a +  b  \cos {\theta})}{8 c} +
  \frac{k_0^3}{2h_0^2} - \frac{3 \mu k_0}{8 h_0^2}  \right) +O(1) \,,
\end{equation}
where $\ell_0=  h_0^{-1}(\tfrac{1}{4} \mu - k_0^2)$ and hence $L$ has a pole of order one. Also we have $M= O(r^{-1})$ and since it is harmonic this implies $M$ also has a
pole of order one so must be of the form $M= m_0/r +O(1)$ where $m_0$
is a constant. Hence $b=0$ so we deduce the triholomorphic Killing field $\partial_\psi = a \partial_{\tilde{\psi}}$.  In summary, so far we have shown that a regular horizon corresponds to a simple pole of the harmonic functions $H,K,L,M$.

We now turn to global constraints.
The
precise periodicities of $(\tilde{\psi}, \tilde{\phi})$ determine the
horizon topology which in general may be that of a lens space.  Now, asymptotic flatness fixes the identifications of the GH space angles $(\psi, \phi)$ to be standard Euler angles on $S^3$ (\ref{vpm}). This will impose identifications on the $(\tilde{\psi}, \tilde{\phi})$ angles. To analyse this, it is convenient to note that the Killing vectors on the horizon which have fixed points at the poles
$\tilde{\theta}=0, \pi$ are $\tilde{v}_\pm = \partial_{\tilde{\phi}}
\mp \partial_{\tilde{\psi}}$. For a lens space $L(p,q)$ these must be related
to the independently $2\pi$-periodic vectors (\ref{vpm}) of the asymptotically flat region by 
\begin{equation}
\left( \begin{array}{c} \tilde{v}_-  \\ \tilde{v}_+ \end{array}
\right) = 
A  \left( \begin{array}{c} v_- \\  v_+ \end{array} \right)   \label{GL2Z}
\end{equation}
where  $A\in GL(2,\mathbb{Z})$ and $\det A=p \in \mathbb{Z}$. The corresponding transformation in terms of Euler angles can be deduced from (\ref{vpm}), which implies $\det \gamma_{ij}|_{i,j= \tilde{\psi},
  \tilde{\phi}} = p^2 \det \gamma_{ij}|_{i,j= {\psi}, {\phi}}$. Thus, comparing to our local analysis above shows that $c=\pm p$ is precisely the integer which defines the lens spaces $L(p,q)$. In fact, by fixing the sign of $p$ appropriately we will identify the constant in \eqref{NHH} as
\be
  h_0 = p  \; .
\ee

Equations \eqref{NHH}--\eqref{NHLM} derived in this section are
necessary conditions for regularity at the horizon.  We will
examine sufficient conditions for a regular horizon in section \ref{sec:horizonreg}.

\subsection{$S^1\times S^2$ horizon}
\label{sec:ringhorizon}

We will now repeat the above analysis for the other type of near-horizon geometry. The horizon data is now
\begin{equation}
\begin{gathered}
\Delta_0 = 0\,, \qquad h|_{\lambda=0} = \frac{R}{\ell} \td \tilde{\psi} \,,\\
\gamma|_{\lambda=0}  = R^2 \td \tilde{\psi}^2 + \ell^2 ( \td \tilde{\theta}^2 + \sin^2 \tilde{\theta} \td \tilde{\phi}^2 )\,,
\end{gathered}
\end{equation}
where the constant $R>0$ has been introduced for later
convenience. Thus, near the horizon the metric is given by
\eqref{gnh} where
\begin{equation}\label{nh2data}
\begin{gathered}
  {F} = \frac{1}{\ell^2} +O(\lambda)\,,\qquad
  h^i =\frac{1}{R\ell}\delta^i_{\tilde{\psi}} +O(\lambda)\,,\\
  h_{\tilde{\theta}} = O(\lambda)\,, \qquad
  \gamma_{\tilde{\theta} \tilde{\theta}} = \ell^2+O(\lambda) \,, \qquad
  \gamma_{\tilde{\theta} i}=O(\lambda) \,, \\
  \gamma_{ij} \td \tilde{\phi}^i \td \tilde{\phi}^j = R^2 \td
  \tilde{\psi}^2 +\ell^2 \sin^2 \tilde{\theta} \td
  \tilde{\phi}^2+O(\lambda)  \,.
\end{gathered}
\end{equation}
The near-horizon geometry has biaxial symmetry generated by the
Killing fields $\partial_{\tilde{\psi}}, \partial_{\tilde{\phi}}$. As
before, these must be related by a constant linear transformation to
the biaxial Killing fields of the GH space, so
$\det \gamma_{ij}|_{i,j= \tilde{\psi}, \tilde{\phi}} = c^2 \det
\gamma_{ij}|_{i,j= {\psi}, {\phi}}$ for some constant $c>0$. Then,
using \eqref{rhonh} and \eqref{znh} we find
\begin{equation}
\rho = \frac{R}{c} \lambda \sin \tilde{\theta}  +O(\lambda^2)\,, \qquad z = \frac{R}{c}\lambda \cos \tilde{\theta}+O(\lambda^2)\,,  \label{rhoznh}
\end{equation}
so the $\mathbb{R}^3$ polar coordinates are
\begin{equation}
r = \frac{R}{c} \lambda + O(\lambda^2)  \,, \qquad \cos\theta =
\cos\tilde{\theta}  + O(\lambda)\,.
\end{equation}
Hence, again, the horizon corresponds to a point in the orbit space metric. 

We may now compare to the orbit space metric \eqref{orbit} and \eqref{orbitnh} near the horizon. Using \eqref{rhoznh} we find
\begin{equation}
\frac{H}{f} = \frac{c^2 \ell^2}{R^2 \lambda^2}+O(\lambda^{-1})  = \frac{\ell^2}{ r^2}+ O(r^{-1})\,.
\end{equation}
The function $f= \lambda^2 \tilde{\Delta} + O(\lambda^3)$, for some non-zero constant $\tilde{\Delta}$~\cite{Reall:2002bh}. Thus we learn that
\begin{equation}
H = \frac{\tilde{\Delta} \ell^2}{R^2} + O(\lambda)\,,
\end{equation}
so in this case $H$ is a smooth harmonic function at the horizon.  

We now determine the other harmonic functions. Near the horizon the
invariant $g_{t\psi} = V \cdot \partial_\psi = \lambda h
\cdot \partial_\psi = \lambda a h_{\tilde{\psi}} + O(\lambda^2)$,
where the final equality follows from writing $\partial_\psi =
a \partial_{\tilde{\psi}}+ b \partial_{\tilde{\phi}}$ for constants
$a$, $b$. Then, expanding \eqref{K}  near the horizon we find
\begin{equation}
KH^{-1} = - \frac{a R}{\ell \tilde{\Delta}\lambda} +O(1)\,,
\end{equation}
where we have used the near-horizon expansion of $f$ and smoothness of $\Psi$. Hence 
\begin{equation}
K = - \frac{a\ell}{R \lambda}+O(1) = -\frac{a \ell}{c r} +O(1)\,.
\end{equation}
Using the invariants $f$ and $V \cdot \partial_\psi=- f^2 \omega_\psi$, together with \eqref{Ldef} and \eqref{Mdef}, then implies
\begin{equation}
L= \frac{(1-a^2) R^2}{\tilde{\Delta} c^2 r^2} +O(r^{-1}), \qquad M= \frac{a(1-a^2) R^4}{\tilde{\Delta} \ell c^3 r^3}+ O(r^{-2})\,.
\end{equation}
Therefore $L$ has a pole of order at most two. However, a harmonic function in $\mathbb{R}^3$ with a pole of order two must be of the form $L= c_1 r^{-2} \cos \theta +c_2 r^{-1} +O(1)$. Thus, since the coefficient of the $r^{-2}$ term is a constant it must vanish and hence we have $a=1$ (choosing a sign). Hence $L=O(r^{-1})$ so harmonicity implies it has a pole of at most order one. This then also implies $M$ has a pole of at most order two. 

In fact we may show that $M=O(r^{-1})$ as follows. The explicit expression for the invariant $g_{\psi\psi}$ in \eqref{invariants}, together the above behaviour of the harmonic functions  requires $K^3 M/ (HL +K^2)^2$ to be smooth at $\lambda=0$.  It follows that
\begin{equation}
M = O(r^{-1})
\end{equation}
and hence harmonicity implies it must also have a pole of order one.  

We now turn to global constraints imposed by asymptotic flatness. In
the case of a $S^1\times S^2$ horizon, $\partial_{\tilde{\psi}}$ and
$\partial_{\tilde{\phi}}$ are independently periodic and we will choose the constant $R$ so that $\tilde{\psi}$ has
period $4\pi$. Hence we can relate $2 \partial_{\tilde{\psi}}$ and
$\partial_{\tilde{\phi}}$ to the independently $2\pi$ periodic vectors
$v_+$, $v_-$ of the asymptotically flat region \eqref{vpm} by a
$SL(2, \mathbb{Z})$ transformation. Their relation to the Euler angles of the Gibbons--Hawking base can then be
deduced from \eqref{vpm}. We find
$\det \gamma_{ij}|_{i,j= \tilde{\psi}, \tilde{\phi}} = \det
\gamma_{ij}|_{i,j= {\psi}, {\phi}}$ so we deduce the constant $c=1$.

We will examine sufficient conditions
for regularity of the horizon in section \ref{sec:horizonreg}.
  
\subsection{Horizon regularity and topology}\label{sec:horizonreg}

As we have seen above a regular horizon corresponds to an isolated
singularity in the $\mathbb{R}^3$ base of the GH space which we may
take to be the origin $r=0$. Furthermore, the harmonic functions have
at most simple poles at the horizon and so can be written as
\begin{equation}
H = \frac{h_0}{r} + H_0\,, \qquad K= \frac{k_0}{r}+ K_0\,, \qquad L =
\frac{\ell_0}{r}+L_0\,, \qquad M = \frac{m_0}{r}+M_0\,, \label{NHharmonicfunc}
\end{equation}
where $h_0, k_0, \ell_0, m_0$ are constants and $H_0, K_0, L_0, M_0$
are harmonic functions that are smooth at $r=0$.  Thus we can write
$H_0= c_0 +O(r)$, where $c_0$ is a constant and the $O(r)$ term is analytic in $r$, and similarly for the other
harmonic functions.  In particular, using \eqref{invariants} this
implies
\begin{equation}
  \begin{aligned}
    \frac{H}{f} &= \frac{\alpha_0}{r^2} + \frac{\alpha_1}{r} +O(1)\,, \\
    g_{\psi\psi}  &= \beta_0+ r \beta_1+ O(r^2)\,, \\ 
    g_{t\psi} &= r( \gamma_0+ r \gamma_1+ O(r^2))\,,
  \end{aligned}
\end{equation}
where $\alpha_i, \beta_i, \gamma_i$ are {\it constants} and comparing
to the near-horizon analysis in the previous sections implies
\begin{equation}
 \alpha_0>0, \qquad \beta_0>0\,.   \label{nhineq}
\end{equation}
Then we can write
\begin{equation}
 f = r\left(  \frac{h_0}{ \alpha_0}  + \frac{c_0\alpha_0-h_0\alpha_1}{\alpha_0^2} r +O(r^2) \right)\,.
\end{equation}
The explicit expressions for leading order coefficients are
\begin{equation}
  \begin{aligned}
    \alpha_0 & = h_0 \ell_0+ k_0^2\,,\\
    \beta_0 & =\frac{ -h_0^2m_0^2 - 3h_0k_0\ell_0m_0 + h_0\ell_0^3 -
      2k_0^3m_0 + \frac{3}{4} k_0^2\ell_0^2}{\alpha_0^2}\,,\\
    \gamma_0^2 &= \frac{\alpha_0-h_0^2\beta_0}{\alpha^2_0} =
    \left(\frac{h_0^2m_0 + \frac{3}{2}h_0k_0\ell_0 + k_0^3}{\alpha_0^2} \right)^2 \,.
  \end{aligned}
\end{equation}
Notice that the last relation shows that $\gamma_0^2\geq 0$ does not lead
to any inequalities beyond (\ref{nhineq}).  In fact, from the very
same relation we can see that $\alpha_0^2\beta_0 >0$ actually implies
$\alpha_0>0$ so \eqref{nhineq} is really equivalent to the single condition
\begin{equation}
  -h_0^2m_0^2 - 3h_0k_0\ell_0m_0 + h_0\ell_0^3 -
  2k_0^3m_0 + \frac{3}{4} k_0^2\ell_0^2 > 0 \label{nhineq_reduced}
\end{equation}
on the parameters $h_0$, $k_0$, $\ell_0$, $m_0$.   It is worth noting that
\be
K^2+HL = \frac{\alpha_0}{r^2} + O(r^{-1}), \qquad g^{tt} = - \frac{\alpha_0 \beta_0}{r^2} + O(r^{-1})
\ee
which already confirms that the above inequalities imply the solution is smooth (\ref{smooth1}) and stably causal (\ref{causality}) near (but not at) the horizon. We will now show that \eqref{nhineq_reduced} is  sufficient for regularity at the horizon. 
 
 To this end, let us perform a coordinate transformation
\begin{equation}\label{NHcoordtransf}
 \dt = \dv + \left(\frac{A_0}{r^2}+\frac{A_1}{r}\right)\dr\,, \qquad
 \dps = \dps' + \frac{B_0}{r}\dr  + C\dph' \,, \qquad \dph=\dph'  \,,
\end{equation}
where $A_0$, $A_1$, $B_0$, $C$ are constants to be determined. Using the
above expansions, it follows that $g_{rr}$ contains $1/r^2$ and $1/r$
singular terms, whereas $g_{r\psi'}$ contains $1/r$ singular
terms. Demanding that the $1/r^2$ term in $g_{rr}$ and $1/r$ term in
$g_{t\psi'}$ vanish is equivalent to setting
\begin{equation}
A_0^2 = \beta_0 \alpha_0^2, \qquad B_0 = - \frac{A_0
  \gamma_0}{\beta_0}\,.
\end{equation}
Demanding that the $1/r$ term in $g_{rr}$ vanishes fixes
\begin{equation}
A_1 = \frac{\alpha_0 \beta_0}{2 A_0} \left( B_0^2 \beta_1+ 2B_0 A_0
  \gamma_1+ \alpha_1- \frac{2 h_0(c_0\alpha_0-h_0\alpha_1)}{\alpha_0^3} A_0^2 \right)\,.
\end{equation}
Note that we have simplified $A_0$, $A_1$ using the identity for
$\gamma_0$ above. With these choices, $g_{rr}$ and $g_{r\psi'}$ are analytic at $r=0$.  

We will also need the near-horizon behaviour of the 1-forms $\chi$,
$\omh$, $\xi$. Using the behaviour of the harmonic functions
\eqref{NHharmonicfunc} near the horizon we find
\begin{align}
  \star_3\d\chi &= \big(-h_0r^{-2} + \Ord(1)\big)\dr + \Ord(r)\dth\,,\\ 
  \star_3\d\omh &= \Ord(r^{-2})\dr + \Ord(1)\dth\,, \\
  \star_3\d\xi &= \big(k_0r^{-2} + \Ord(1)\big)\dr + \Ord(r)\dth\,,
\end{align}
and writing the 1-forms as (\ref{1forms}), we may integrate to get
\begin{equation}
  \chi = \big(h_0\cos\theta + \tilde{\chi}_0 + \Ord(r^2)\big)\dph\,, \qquad 
  \omh = \Ord(1)\dph\,, \qquad
  \xi = \big(-k_0\cos\theta + \tilde{\xi}_0 + \Ord(r^2)\big)\dph
\end{equation}
for some constants $\tilde{\chi}_0$, $\tilde{\xi}_0$. For convenience we choose
$C=-\tilde{\chi}_0$ in \eqref{NHcoordtransf}.  The full metric near $r=0$ now
reads
\begin{multline}\label{fullmetricNH}
  \ds^2 = -r^2\Big( \frac{h_0^2}{\alpha_0^2} + \Ord(r) \Big) \big(\td
  v+\Ord(1)\dph\big)^2 \pm  2  \Big(
    \frac{1}{\sqrt{\beta_0}} + \Ord(r) \Big) \big(\td v+ \Ord(1)\dph \big)\td
    r  +O(1) \td r^2 \\
  + \Ord(1) \big( \td \psi'+ h_0 \cos \theta \td \phi' +\Ord(r^2)
  \big) \td r
  + 2r\big(  \gamma_0 + \Ord(r)\big)\big(\td v+ \Ord(1)\dph \big)\big( \td \psi'+ h_0 \cos
  \theta \td \phi'\big)  \\
  + \big( \beta_0 + \Ord(r)\big) \big( \td \psi'+ h_0 \cos \theta \td \phi'
  +\Ord(r^2)\big) ^2 + \big(\alpha_0 + \Ord(r)\big) \big(\dth^2+ \sin^2 \theta
  \td \phi'^2\big)\,.
\end{multline}
The metric and its inverse are now analytic at $r=0$, hence the
spacetime can be analytically extended to the region $r \leq 0$. The
surface $r=0$ is an extremal Killing horizon with respect to the
supersymmetric Killing field $V = \partial /\partial v$. The upper (lower) sign in $g_{vr}$ corresponds to a future (past) horizon. The gauge
field near $r=0$ is given by
\begin{multline}
  A = \frac{\sqrt{3}}{2}\bigg[\Big(\frac{h_0}{\alpha_0}r + \Ord(r^2)\Big)\dv 
  \pm
  \Big(\frac{\beta_0h_0-\gamma_0(h_0m_0+\frac{1}{2}k_0\ell_0)}{\sqrt{\beta_0}r
   }+ \Ord(1)\Big)\dr\\
  + \Big(\frac{h_0m_0+\frac{1}{2}k_0\ell_0}{\alpha_0} +
  \Ord(r)\Big)(\dps' + h_0\cos\theta\dph') - \Big(\tilde{\xi}_0 - k_0\cos\theta +
  \Ord(r)\Big)\dph'\bigg]\,,
\end{multline}
so we see that the only singular terms are pure gauge. Therefore the Maxwell field $F=\d A$ (and hence the full solution)
is analytic at $r=0$.

The near-horizon limit can be taken by
transforming to coordinates $(v,r)\to
(v/\epsilon,\epsilon r)$ and letting $\epsilon \to 0$, 
giving the near-horizon geometry
\begin{multline}\label{NHmetric}
    \td s^2_{\text{NH}} = -r^2 \frac{h_0^2}{\alpha_0^2} \dv^2  \pm
    \frac{2}{\sqrt{\beta_0}}  \dv \dr 
  + 2r \gamma_0 \dv( \td \psi'+ h_0 \cos \theta \td \phi')  \\
    +  \beta_0  ( \td \psi'+ h_0 \cos \theta \td \phi') ^2 +
    \alpha_0  (\d \theta^2+ \sin^2 \theta \td \phi'^2)\,,
\end{multline}
as well as the near-horizon Maxwell field
\begin{equation}
  F_{\text{NH}} =  \frac{\sqrt{3}}{2}\bigg[
\frac{h_0}{\alpha_0}\dr\wedge\dv 
  -
  \Big(\frac{h_0}{\alpha_0}(h_0m_0+\tfrac{1}{2}k_0\ell_0)+k_0\Big)\sin\theta\dth\wedge\dph'
  \bigg]\,.
\end{equation}
The second line in \eqref{NHmetric} is the metric induced on cross-sections of the horizon.
For $h_0=0$ it is simply the standard product metric on $S^1\times S^2$.
For $h_0 \neq 0$ it is a locally homogeneous metric on $S^3$.

Our analysis so far in this section has been local. We will now examine the constraints imposed by asymptotic flatness.  Recall in our near-horizon analysis in section \ref{sec:S3horizon} and \ref{sec:ringhorizon} we showed that $h_0=p \in \mathbb{Z}$ is the integer which fixes  the horizon topology to be a lens space $L(p,q)$. The precise topology is determined by the identifications on the angles. These are already fixed by asymptotic flatness which requires $(\psi, \phi)$ to be identified as the standard Euler angles on $S^3$.  For $p \neq 0$, it is convenient to define $\bar{\psi} = \psi' /p$ and $\bar{\phi}=\phi'$.  From the coordinate change \eqref{NHcoordtransf}, the Killing fields are related by 
\be
\begin{pmatrix} \partial_{\bar{\psi}}
  \\ \partial_{\bar{\phi}} \end{pmatrix} =
\begin{pmatrix} p & 0 \\ -\tilde{\chi}_0 & 1 \end{pmatrix}
\begin{pmatrix} \partial_{{\psi}} \\ \partial_{{\phi}} \end{pmatrix}
\ee
and hence the matrix $A$ in (\ref{GL2Z}) is determined using (\ref{vpm}). Requiring the entries of $A$ to be integer is then equivalent to $\tilde{
\chi}_0=p+2n -1$ for some integer $n$. The matrix $A$ then simplifies to
\be
A =  \begin{pmatrix} 1-n & n \\ -p-n+1 & p+n \end{pmatrix} \,.
\ee
By a basis change $A\to A'= A B$ where $B\in SL(2, \mathbb{Z})$ we can bring the matrix into triangular form
\be
A' =  \begin{pmatrix} 1 & 0  \\ q & p \end{pmatrix} \,,
\ee
where 
\be
B = \begin{pmatrix} \alpha & -n \\ \beta & 1-n \end{pmatrix}
\ee
and $(1-n) \alpha+ n \beta=1$ and $q = 1 +p (\beta- \alpha)$. We deduce the important result
\be
q \equiv 1 \mod p \,.
\ee
Therefore, we have shown that the identifications that arise from asymptotic flatness, together with regularity, imply the only possible horizon topology is $L(p, 1)$.  With these global identifications we find the area of cross-sections of the horizon is
\begin{equation}
  A = 16\pi^2 \sqrt{-h_0^2m_0^2 - 3h_0k_0\ell_0m_0 + h_0\ell_0^3 -
  2k_0^3m_0 + \frac{3}{4} k_0^2\ell_0^2 } \; ,
\end{equation}
where the expression under the square root is positive (\ref{nhineq_reduced}).
This completes our analysis of the horizon.

\subsection{Summary}
To summarise, we have established the following results.

\begin{theorem} \label{thm:horizon} Consider a supersymmetric and
  biaxisymmetric solution to minimal supergravity containing a smooth
  supersymmetric horizon with compact cross-sections of topology
  $S^1 \times S^2$ or locally $S^3$. In the orbit space metric the
  horizon is an isolated singular point on the boundary $\rho=0$,
  which we may take to be the origin $\rho=z=0$. Equivalently, in the
  Gibbons--Hawking metric, the horizon is an isolated singular point on
  the $z-$axis, which we may take to be the origin of
  $\mathbb{R}^3$. Furthermore, the harmonic functions can be written
  as
\begin{equation}
H = \frac{h_0}{r}+ H_0  \,,\qquad
K = \frac{k_0}{r} + K_0  \,,\qquad
L = \frac{\ell_0}{r} + L_0  \,,\qquad
M = \frac{m_0}{r} +M_0 \,,
\end{equation}
where $r= \sqrt{\rho^2+ z^2}$, $H_0$, $K_0$, $L_0$, $M_0$ are harmonic
functions which are smooth at $r=0$ and $h_0$, $k_0$, $\ell_0$, $m_0$
are constants, where $h_0 \neq 0 $
for a locally $S^3$ horizon and $h_0 =0$ for a $S^1\times S^2$
horizon. In addition, the parameters satisfy
\begin{equation}
  -h_0^2m_0^2 - 3h_0k_0\ell_0m_0 + h_0\ell_0^3 -
  2k_0^3m_0 + \frac{3}{4} k_0^2\ell_0^2 > 0\,,
\end{equation}
which in particular also implies that $h_0\ell_0 + k_0^2
>0$.
\end{theorem}

\begin{theorem}\label{thm:horizon+AF} Consider an asymptotically flat,
  supersymmetric and biaxisymmetric black hole solution to
  five-dimensional minimal supergravity.
  \begin{enumerate}
  \item Cross-sections of any connected
  component of the horizon must be homeomorphic to $S^3$, $S^1\times
  S^2$ or a lens space $L(p,1)$.  
  \item The coefficient of the
  singular term in the harmonic function $H$ is $h_0=\pm 1$ for an
  $S^3$ horizon, $h_0=0$ for $S^1\times S^2$ and more generally $h_0=p \in \mathbb{Z}$ for $L(p,1)$.
  \end{enumerate}
\end{theorem}

\noindent {\bf Remarks}
\begin{enumerate}
\item Theorem \ref{thm:horizon} is a five-dimensional analogue of Theorem 3.2 in~\cite{Chrusciel:2005ve} which is a crucial ingredient for the classification of supersymmetric four-dimensional black hole spacetimes.  
\item The horizon is an isolated singular point of the orbit space metric also in the vacuum case. In fact this is a consequence of the $SO(2,1)$-symmetry of near-horizon geometries~\cite{Figueras:2009ci}.
\item We will offer an alternative proof of part 1 of Theorem \ref{thm:horizon+AF} by analysing the rod structure of the general solution, see section \ref{sec:rodstructure}.
\end{enumerate}

\section{Geometry and topology of the axes}
\label{sec:axes}

\subsection{Rod structure}
\label{sec:rodstructure}

We now analyse the axes $\mathcal{A}$ in more detail. Recall that the
axes is the part of the boundary of the orbit space $\hat{M}$ where
$\det \gamma_{ij} =0$, i.e.~the $U(1)^2$-symmetry has fixed points. In
Weyl coordinates $(\rho,z)$ the boundary $\partial \hat{M}$ is the
$z$-axis. As shown above, a horizon corresponds to an isolated
singular point on the $z$-axis. The
remaining part of the $z$-axis corresponds to $\mathcal{A}$, which
splits into intervals along which $\gamma_{ij}$ is of rank $1$ with
endpoints on which $\gamma_{ij}$ is of rank $0$. The intervals where
$\gamma_{ij}$ is rank-$1$ correspond to the axis rods, and the
endpoints where $\gamma_{ij}=0$ to the corners of the orbit space
$\hat{M}$.

\begin{lemma}\label{lemma:axis}
  Smoothness of the spacetime near an axis rod $I$ implies
  $\omhph=O(\rho^2)$ and $\chi = \chi|_I + O(\rho^2)$ where
  $\chi|_I$ is an odd integer.  The Killing field which vanishes on
  $I$ is
  \begin{equation}\label{v}
    v=\partial_\phi - \chi|_I \partial_\psi\,,
  \end{equation}
  where the normalisation has been fixed so the orbits of $v$ (away
  from $I$) are $2\pi$-periodic.
\end{lemma}

\begin{proof}
  First note that $\omh$ and $\chi$ may be
  related to invariants of the metric as
  \begin{equation}
    \omhph = \frac{1}{N}\begin{vmatrix} g_{t\psi} &
      g_{t \phi} \\ g_{\psi\psi} & g_{\psi \phi} \end{vmatrix}\,,
    \qquad
    \chiph = - \frac{1}{N} \begin{vmatrix} g_{t t} &
      g_{t \psi} \\ g_{t \phi} & g_{\psi \phi} \end{vmatrix} \,.
  \end{equation}
  These are smooth wherever $N>0$, which, by Lemma \ref{lemma:orbitspace}, includes
  $I$. Now, the $2\pi$-periodic Killing field which vanishes on $I$
  must be of the form
  \begin{equation}
    v = a v_+ + b v_- = (b-a) \partial_\psi+ (a+b) \partial_\phi\,,
  \end{equation}
  where the coefficients $(a,b)$ relative to the $2\pi$-periodic basis
  $\{v_+,v_-\}$ are coprime integers (not both vanishing), and where
  in the second equality we have used \eqref{vpm}. Furthermore,
  $a+b \neq 0$, as otherwise $\partial_\psi=0$ which we know cannot occur
  on $I$, again by Lemma \ref{lemma:orbitspace}. Therefore, on $I$ we can write
  $\partial_\phi = c \partial_\psi$, where $c = (a-b)/(a+b)$, and
  hence $g_{t\phi}=c g_{t\psi}$ and $g_{\psi \phi} = c g_{\psi \psi}$.
  It follows that on $I$ we have $\omhph=0$ and
  $\chiph = c$ and therefore the Killing field vanishing on $I$ is
  $v = (a+b) (\partial_\phi - \chi|_I \partial_\psi)$, where
  $\chi|_I=c$.

  To determine the behaviour of $\omhph$ and $\chiph$ near $I$
  we argue as follows. First, from \eqref{1forms}, equation
  \eqref{dchi} is equivalent to
  $\partial_z\chiph = - \rho \partial_\rho H$ and
  $\partial_\rho \chiph = \rho \partial_z H$.  We know the
  axisymmetric harmonic function $H$ is smooth at $I$, so near $I$ we
  can write $H= H_0(z)+ O(\rho^2)$ for some smooth function
  $H_0(z)$. Therefore, integrating it follows that
  $\chiph = \chi|_{I} + O(\rho^2)$ where $\chi|_{I}$ is a
  constant. Similarly, integrating equation \eqref{domegahat} for
  $\omh$ and using the fact that the harmonic functions
  $K,L,M$ are also smooth at $I$, implies
  $\omhph= \hat{\omega}|_I+ O(\rho^2)$ where
  $\hat{\omega}|_I$ is a constant. Comparing to the above we deduce
  the constant $\hat{\omega}|_I=0$.

  Putting things together we find the spacetime metric for $z\in I$
  and $\rho \to 0$ is
  \begin{multline}
    \d s^2|_{\text{near } I} = - f^2 \big( \dt + O(\rho^2) \dph_I\big)^2 +
    \frac{H}{f} \big( \d\rho^2 + \rho^2 \dph_I^2+ \d z^2\big) +
    g_{\psi\psi}\big( \dps_I + O(\rho^2) \dph_I \big)^2 \\
    +2 g_{t\psi} \big( \dt + O(\rho^2) \dph_I\big)\big(\dps_I+
    O(\rho^2) \dph_I\big)\,,
  \end{multline}
  where we have defined new coordinates
  $(\psi_I, \phi_I) =( \psi+\chi|_I \phi, \phi)$. Now, using
  smoothness of the harmonic functions, (\ref{invariants},\ref{I}) implies the
  invariants  $f$, $g_{\psi\psi}$, $g_{\psi t}, N$ must also be smooth on
  $I$, and near $I$ the corrections must be $O(\rho^2)$. Furthermore,
  by (\ref{H}) and Lemma \ref{lemma:orbitspace}, $H/f>0$ and $g_{\psi\psi}>0$ on $I$. Hence, the above
  is a smooth Lorentzian metric iff the angles are identified as
  $(\psi_I, \phi_I+2\pi)\sim (\psi_I,\phi_I)$ (this can be seen by
  converting to cartesian coordinates in the $(\rho, \phi_I)$ plane).
  On the other hand, the identifications on the Euler angles
  $(\psi, \phi)$ from asymptotic flatness (\ref{vpm}) imply the new angles are identified as
  $(\psi_I+4\pi, \phi_I)\sim (\psi_I,\phi_I)$ and
  $(\psi_I+ 2\pi(1+\chi|_I), \phi_I+2\pi)\sim (\psi_I,
  \phi_I)$. Compatibility of these periodicity lattices requires that
  $\chi|_I$ is an odd integer.

  Now, in terms of the new coordinates $v =
  (a+b) \partial_{\phi_I}$. We have just seen, however, that
  smoothness at the axis requires $\partial_{\phi_I}$ to have
  $2\pi$-periodic orbits, so that we must have $a+b=1$. Hence in the
  original coordinates the Killing field which vanishes on $I$ is
  indeed given by \eqref{v} and $\chi|_I = 2a-1$ is an odd integer.
\end{proof}

Let us now denote the axis rods by $I_i = (z_i, z_{i+1})$ for
$i=1, \dots, n-1$, where $z_1<z_2< \dots < z_{n}$, and
$I_- = (-\infty, z_1)$ and $I_+ = (z_{n}, \infty)$.  As we have just
established, the $2\pi$-normalised Killing fields vanishing on the
respective axis rods are
\begin{equation}\label{rodvecs_euler}
  v_i = \partial_\phi- \chi_i \partial_\psi, \qquad
  v_\pm = \partial_\phi  - \chi_\pm \partial_\psi  \,.
\end{equation}
where $\chi_i \equiv \chi|_{I_i}$.  The data
$\{ (I_i, v_i) \; | \; i,j=+, -, 1,..,n-1 \}$ defines the rod structure
of the spacetime~\cite{Hollands:2007aj}.

There are certain compatibility requirements between adjacent rods
that have been derived for stationary and biaxisymmetric
spacetimes~\cite{Hollands:2007aj}: If $v_i$ and $v_j$ are the
$2\pi$-normalised rod vectors of adjacent axis rods,
$\det ( v^T_i v^T_j)=\pm 1$. If two axis rods, with vectors $v_i$ and $v_j$, are separated only by a horizon, then
$\det ( v^T_i v^T_j)=p \in \mathbb{Z}$ and the topology of the
horizon is $S^1\times S^2$ for $p=0$, $S^3$ for $p=\pm 1$, and in
general a lens space $L(p,q)$ where $q\in \mathbb{Z}$ is only defined modulo $p$.

For asymptotically flat solutions
$\chi_\pm =\pm 1$ and so $v_\pm$ coincide with \eqref{vpm}, thus
defining a natural $2\pi$-normalised basis. In this basis the rod
vectors \eqref{rodvecs_euler} are
\begin{equation}\label{rodvecs}
  v_-=(1,0)\,,\qquad
  v_i = (1-a_i, a_i)\,,\qquad
  v_+ = (0,1)\,,
\end{equation}
where by Lemma \ref{lemma:axis} we have defined
$a_i \equiv \tfrac{1}{2}(1+\chi_i) \in \mathbb{Z}$ for
$i=1, \dots, n-1$.  The determinants of adjacent rod vectors are then
given by ($i=1, \dots, n-1$)
\begin{equation}\label{detrodvecs}
  \det ( v^T_- \; v^T_1)  =a_1\,, \qquad
  \det ( v^T_i  \; v^T_{i+1} ) = a_{i+1}-a_i\,,\qquad
  \det ( v^T_n \;  v^T_+) =1-a_n\,.
\end{equation}
Evidently, the rod structure is somewhat restricted. In particular,
this provides extra constraints on the horizon topology, thus
providing an alternative proof of Theorem \ref{thm:horizon+AF} (part 1) as
follows.

\begin{proof}[Proof of Theorem \ref{thm:horizon+AF}, part 1] The proof is elementary.  If $z=z_1$ is a horizon, the
  rod vectors which vanish on either side of the horizon are $v_-$ and
  $v_1$, therefore \eqref{detrodvecs} implies the horizon topology is
  $L(a_1, 1)$. Similarly, if $z=z_{n}$ is a horizon its topology is
  $L(1-a_n,1)$.  If $z=z_i$ is a horizon for some $i=2, \dots, n-1$ the
  rod vectors which vanish on either side are $v_{i-1}$ and
  $v_i$. Thus defining
 \begin{equation}
   P\equiv \left( v_{i-1}^T v_i^T \right) = 
   \begin{pmatrix} 1-a_{i-1} & 1- a_i \\ a_{i-1} & a_i \end{pmatrix}
 \end{equation}
 we find $\det P =a_i-a_{i-1}= p$ where $p \in \mathbb{Z}$ and the
 horizon topology is $L(p,q)$. We can obtain $q$ by a performing a
 basis change $P \to P'=AP$ where $A \in SL(2, \mathbb{Z})$ to put
 this into standard form
 \begin{equation}
   P' = \begin{pmatrix} 1 & q \\ 0 &   p \end{pmatrix}\,.
 \end{equation}
 We find
 \begin{equation}
   A = \begin{pmatrix} a & b \\ -a_{i-1} & 1-a_{i-1} \end{pmatrix}
 \end{equation}
 where the unit
 determinant condition is $1= a(1-a_{i-1})+b a_{i-1}$.  Therefore
 \begin{equation}
   q= a(1-a_i)+b a_i = 1+(b- a)(a_i-a_{i-1})  \equiv 1 \mod p \,,
 \end{equation}
 where in the second equality we used the unit determinant condition.
 This shows the horizon topology is $L(p, 1)$ as claimed. If $p=\pm 1$
 this is just $S^3$, and if $p=0$ this is $S^1\times S^2$.
\end{proof}

\subsection{Geometry of the axes}

We now turn to the analysis of the metric on the axes. Consider an axis rod $I_i=(z_i, z_{i+1})$ where $z=z_i$ is a
corner of the orbit space. The induced metric on $I_i$ is
\begin{equation}
  \ds^2|_{I_i} = -f^2 \dt^2  + \frac{H}{f} \d z^2+ g_{\psi\psi}
  \dps_i^2 + 2g_{t \psi} \dt \dps_i\,,
\end{equation}
where $\psi_i=\psi+\chi_i \phi$. This is a 3-dimensional timelike
submanifold with a circle action generated by
$\partial_\psi = \partial_{\psi_i}$ which has a fixed point at the
endpoint $z=z_i$. At this fixed point we must have
$g_{\psi\psi}=g_{t\psi}=0$.  We will now analyse the conditions
required by smoothness of the geometry near such a fixed point
$z=z_i$.

It is convenient to use as a coordinate the proper distance from the
fixed point
\begin{equation}
  s = \int_{z_i}^z \sqrt{g_{zz}} \d z.
\end{equation}
Now, smoothness at $z=z_i$
requires that $4 (\td | \partial_{\psi_i} |)^2 \to 1$ as $z\to
z_i^+$ (recall that $\Delta \psi_i = 4\pi$). In terms of the proper distance this is equivalent to
\begin{equation}
  g_{\psi\psi} = \tfrac{1}{4} s^2 (1+ O(s^2))\,,
\end{equation}
where the subleading terms are fixed by smoothness at $s=0$ (and
converting to cartesian coordinates in the $(s,\psi_i)$-plane).
Smoothness of the metric and its inverse on the axis thus also
requires that
\begin{equation}
g_{t\psi}= O(s^2), \qquad f^2 = f_i^2+O(s^2)\,,
\end{equation}
where $f_i = f|_{z=z_i}\neq 0$ by Lemma \ref{lemma:orbitspace}.  This
behaviour of the metric near $s=0$, together with
\eqref{H}, gives $f/H = \tfrac{1}{4}f_i^2 s^2 +O(s^4)$, hence using
$g_{zz}= H/f$ we find
\begin{equation}
  z-z_i = \int_0^s \sqrt{\frac{f}{H} }\ds 
  = \tfrac{1}{4} |f_i| s^2 (1+O(s^2))\,.
\end{equation}
We deduce that
\begin{equation}
  H = \frac{\text{sgn}( f_i)}{z-z_i} + O(1)   \label{Hcorner}
\end{equation}
as $z \to z_i^+$.  But, by Lemma \ref{lemma:orbitspace}, $H$ is a harmonic function on $\mathbb{R}^3$
with an isolated singularity at $(\rho,z)=(0,z_i)$. Thus, \eqref{Hcorner}
implies that the singularity of $H$ is a pole of order one. Therefore,
for $\rho \geq 0$, we must have
\begin{equation}
  H = \frac{h_i}{\sqrt{\rho^2+(z-z_i)^2}} + \Ht_i\,,
\end{equation}
where $h_i=\text{sgn}( f_i)$ and $\Ht_i$ is a harmonic function on
$\mathbb{R}^3$ which is smooth at $(\rho,z)=(0,z_i)$.

We will now determine the behaviour of the other harmonic functions.
Since $f_i \neq 0$, equation \eqref{K} implies that $K H^{-1}$ is also
smooth at any corner.  We deduce that
\begin{equation}
K = \frac{k_i}{\sqrt{ \rho^2+ (z-z_i)^2}} +\Kt_i\,,
\end{equation}
where $k_i$ is a constant (possibly vanishing) and $\Kt_i$ is a
harmonic function smooth at the corner  $(\rho,z)=(0,z_i)$. Next,
smoothness of $f^{-1}$ at the corner and \eqref{Ldef} then implies $L$
may also have a pole of order one at the corner, so
\begin{equation}
  L = \frac{\ell_i}{\sqrt{ \rho^2+ (z-z_i)^2}} + \Lt_i\,,
\end{equation}
where $\ell_i= -h_i^{-1} k_i^2$  and $\Lt_i$ a harmonic function smooth
at the corner. Finally, the invariant $-f^{-2}V \cdot \partial_\psi =
\omega_\psi$ must be smooth at any fixed point of $\partial_\psi$
(since $f_i \neq 0$), and thus \eqref{Mdef} implies
\begin{equation}
  M = \frac{m_i}{\sqrt{ \rho^2+ (z-z_i)^2}} + \Mt_i\,,
\end{equation}
where 
\begin{equation}
m_i= - k_i^3 - \tfrac{3}{2} h_i^{-1}k_i \ell_i = \tfrac{1}{2} k_i^3
\end{equation}
and $\Mt_i$ is a harmonic function smooth at the corner. There are
further conditions arising from the fact $\omega_\psi$ must also
vanish at the corner which we will explore in more detail below.

Therefore, we have found that the boundary conditions arising from
smoothness of the solution on the axes are sufficient to
determine its functional form near any corner of the orbit space.  
The obtained conditions on the solution are \emph{necessary}
conditions for smoothness at a corner of the orbit space. In the following section we will show that in fact they are also sufficient.

\subsection{Smoothness at corners of orbit space}

In this section we complete the smoothness analysis at the corners of
the orbit space.  To do so, let us introduce $\mathbb{R}^3$-polar
coordinates $(r, \theta, \phi)$ centred at the corner
$(\rho,z)=(0,z_i)$, where for notational simplicity we have dropped
the label $i$ in the new coordinates. Then, as just shown,
we can write 
\begin{equation}\label{HKLM_at_centre}
  H = \frac{h}{r} + \Ht\,,\qquad
  K = \frac{k}{r} + \Kt\,,\qquad
  L = \frac{\ell}{r} + \Lt\,,\qquad
  M = \frac{m}{r} + \Mt\,,
\end{equation}
where $h= \pm 1$, $\ell = -h^{-1} k^2$, $m = k^3/2$ and $\Ht$, $\Kt$, $\Lt$,
$\Mt$ are axisymmetric harmonic functions that are smooth at the
centre. Thus we can write
\begin{equation}\label{legendre}
  \Ht = \sum_{l=0}^\infty h_l r^l P_{l}(\cos \theta)
\end{equation}
where $P_l$ are the Legendre polynomials and $h_l$ are constants, and similarly for $\Kt$, $\Lt$,
$\Mt$, where furthermore the constants $h_l$, $k_l$, $\ell_l$, $m_l$ are
such that $\omega_\psi|_{r=0} =0$. We will now show that for asymptotically flat solutions the above conditions are also sufficient for
smoothness at $r=0$, provided that on the adjacent rods $\chi$ is an odd integer and $\omh=0$ (as required by Lemma~\ref{lemma:axis}).

Given \eqref{HKLM_at_centre}, \eqref{legendre} and (\ref{1forms}), we may solve
(\ref{dchi}) for the 1-form $\chi$, giving
\begin{equation}
  \chi  = (h \cos \theta + \chi_0) \dph +\tilde{\chi}\,,
\end{equation}
where $\chi_0$ is a constant and, using the fact that $P_l$ is a
Legendre polynomial,
\begin{equation}
  \tilde{\chi} = r^2 \sin^2 \theta \sum_{l=1}^\infty \frac{h_l
  r^{l-1}}{l+1} P'_{l}(\cos\theta) \dph\,.
\end{equation}
Now, define new  coordinates $(R, \psi', \phi')$ by 
\begin{equation}
  r = \tfrac{1}{4} R^2, \qquad \psi' = \psi+ \chi_0 \phi, \qquad \phi' =
  h \phi\,,
\end{equation}
so that the GH base is
\begin{equation}
  \ds^2_{GH} =  F \left( \d R^2 + \tfrac{1}{4} R^2 \left[ (\dth^2
      + \sin^2\theta \dph^2) +\frac{1}{F^2} (\dps'  +
      \cos \theta \dph' +\tilde{\chi})^2 \right] \right)
\end{equation}
where we have defined $F \equiv r H = h+ r \tilde{H}$.  In terms of the new
coordinates
\begin{equation}
  F = \pm 1 + O(R^2)\,, \qquad \tilde{\chi}  = O(R^4) \td \phi'\,,
\end{equation}
so we see that as $R\to 0$, the GH base approaches the origin
$\pm \mathbb{R}^4$, \emph{if} the angles $(\psi', \phi')$ are
identified as Euler angles on $S^3$.  Since the original angles
$(\psi, \phi)$ are required to be Euler angles on $S^3$ by asymptotic
flatness, it is easy to see that $(\psi', \phi')$ are also Euler
angles on $S^3$ if and only if $\chi_0= h - 2n -1$ for some
$n \in \mathbb{Z}$.  In fact, this condition follows from Lemma
\ref{lemma:axis}: On the axis $\theta = 0, \pi$ it is clear that
$\tilde{\chi}=0$ and thus on any axis rod $I$ we have
$\chi|_{I} = \chi_0\pm h$.  But by Lemma \ref{lemma:axis},  we know that
$\chi|_{I}$ is an odd integer and hence $\chi_0$ is an even integer as
required.

In order to verify the GH base at the centre is actually smooth
requires us to control the higher order terms more carefully. To this
end introduce coordinates\footnote{The coordinates $\phi^\pm$ in this section are different to those in (\ref{vpm}). We will only use these in this section, so there should be no confusion.}
\begin{equation}
  \phi^\pm = \tfrac{1}{2} (\psi' \pm \phi'), \qquad X_+  = R \cos
  (\tfrac{1}{2} \theta), \qquad X_- = R\sin (\tfrac{1}{2} \theta)
\end{equation}
so 
\begin{equation}
  \ds^2(\mathbb{R}^4)  = \d X_+^2+ X_+^2 (\dph^+)^2 + \d X_-^2+ X_-^2 (\dph^-)^2 
\end{equation}
and $\phi^\pm$ are independently $2\pi$ periodic. In these coordinates
any smooth biaxisymmetric function on $\mathbb{R}^4$ is a smooth
function of $(X_+^2, X_-^2)$.  Noting that
\begin{equation}
  r = \tfrac{1}{4}(X_+^2+X_-^2), \qquad r \cos \theta = \tfrac{1}{4} (X_+^2-X_-^2)
\end{equation}
and using the fact that $P_l$ are polynomials of order $l$, it is easy
to see that $\Ht$ and hence $F$ are analytic functions of
$(X_+^2, X_-^2)$. Similarly we find that
\begin{equation}\label{chitilde}
  \tilde{\chi} = \tfrac{1}{4} X_+^2 X_-^2 ( h_1 + \dots ) h \td \phi' =
  \tfrac{1}{4} X_+^2 X_-^2 ( h_1 + \dots ) h (\td \phi^+ - \td \phi^-)  \,
  ,
\end{equation}
where the higher order terms are analytic in $(X_+^2, X_-^2)$, so the
1-form $\tilde{\chi}$ is analytic at the origin of $\mathbb{R}^4$.
Putting everything together, we can write the GH base as
\begin{multline}
\td s^2_{GH} = F \td s^2 (\mathbb{R}^4) - \frac{\tilde{H} ( h +F) }{4F} (X_+^2 \td \phi^+ + X_-^2 \td \phi^-)^2 \\
 + \frac{1}{F} (X_+^2 \td \phi^+ + X_-^2 \td \phi^-) \tilde{\chi} + \frac{F (X_+^2+X_-^2)}{4} \tilde{\chi}^2
\end{multline}
which is now manifestly analytic at the origin of $\mathbb{R}^4$.
Therefore, the GH base is indeed smooth, in fact analytic, at any
centre corresponding to a corner of the orbit space.

We now turn to the other components of the spacetime metric, namely
the function $f$ and 1-form $\omega$. Expanding the regular parts
$\tilde{K}, \tilde{L}, \tilde{M}$ of the harmonic functions $K, L, M$
as above for $\tilde{H}$ it is easy to see that $f$ is an analytic
function of $(X_+^2, X_-^2)$. Recall by Lemma \ref{lemma:orbitspace} we must have $f \neq 0$ at any centre corresponding to a corner of the orbit space.

 It remains to be checked that also the
$1$-form $\omega$ is smooth at the centre. In the above coordinates we
can write
\begin{equation}\label{omegaXYcoord}
  \omega = \omega_\psi (\td \psi'+ \cos \theta \td \phi')+ \omega_\psi
  \tilde{\chi}+ \omh = \frac{2 \omega_{\psi}}{R^2} (X_+^2 \td
  \phi^+ + X_-^2 \td \phi^-) + \omega_\psi \tilde{\chi}+ \omh \,.
\end{equation}
Using \eqref{Mdef} and expanding
\begin{equation}\label{1/H_centre}
  \frac{1}{h+r\Ht} = h - r\Ht + r^2G_1\,,
\end{equation}
where $G_1= \Ht^2/(h+r\Ht)$ is analytic in  $(X_+^2, X_-^2)$, as well as
making use of the identities $h^2=1$, $\ell = -hk^2$, $m = k^3/2$, one
finds
\begin{equation}
  \omega_\psi = \sum_{l=0}^{\infty} \left(m_l - hmh_l  + \frac{3}{2}( hk\ell_l -
   h\ell k_l) \right)r^l P_l(\cos\theta) + r \tilde{G_1}\,,
\end{equation}
where $\tilde{G_1}$ is some analytic function in $(X_+^2, X_-^2)$.  Thus
$\omega_\psi$, and hence also $\omega_\psi \tilde{\chi}$ are analytic
in $(X_+^2, X_-^2)$ and for smoothness of \eqref{omegaXYcoord} we
therefore only need to check that
\begin{equation}\label{cond_smooth}
  \left( \frac{2 \omega_{\psi} X_+^2 }{R^2} +h \omhph_\phi
  \right)\td \phi^+ + \left( \frac{2 \omega_{\psi} X_-^2 }{R^2} - h
    \omhph_\phi \right)\td \phi^-
\end{equation}
is smooth at the origin, or equivalently that
\begin{equation}\label{cond_final}
 \frac{2\omega_\psi X_\pm^2}{R^2} \pm h \omhph_\phi =  X_\pm^2 G_\pm,
\end{equation}
for some smooth functions $G_\pm$ of  $(X_+^2, X_-^2)$.

In fact one can solve equation \eqref{domegahat} for the 1-form $\omh$, of the form (\ref{1forms}), as
\begin{equation}
  \omhph_\phi = \omh_0 + r \sin^2\theta
  \sum_{l=1}^{\infty} 
         \frac{(hm_l-h_lm)+\frac{3}{2}(k\ell_l-k_l\ell)}{l}r^{l-1}
         P'_l(\cos\theta) + hr^2\sin^2\theta G_2\,,
\end{equation}
where we have used that $\omega_\psi|_{r=0} = 0$ and defined
\begin{equation}
  G_2 = h\sum_{j=0}^{\infty}\sum_{l=1}^{\infty} 
        \frac{h_jm_l-h_lm_j +
        \frac{3}{2}(k_j\ell_l-k_l\ell_j)}{l+j+1}r^{l+j-1}
        P'_l(\cos\theta)P_j(\cos\theta)
\end{equation}
which is an analytic function of  $(X_+^2, X_-^2)$.
From Lemma \ref{lemma:axis}, we must have $\omh_0 = 0$.  Then
\begin{multline}\label{cond_smoothX}
  \frac{2\omega_\psi X_\pm^2}{R^2} \pm h \omhph_\phi  =  \frac{X_\pm^2}{2}\tilde{G_1} \pm \frac{X_+^2X_-^2}{4} G_2+\\
  h \sum_{l=1}^{\infty} \left(hm_l - mh_l  + \frac{3}{2}( k\ell_l -
   \ell k_l) \right) r^l \left( (1 \pm \cos\theta)P_l(\cos\theta) \pm
    \frac{\sin^2\theta}{l}P'_l(\cos\theta) \right), 
\end{multline}
and it is obvious that the first two terms are of the required
form. Using basic properties of Legendre polynomials we can rewrite
\begin{align}
  (1\pm\cos\theta)P_l(\cos\theta) \pm
  \frac{\sin^2\theta}{l}P'_l(\cos\theta) =
  P_l(\cos\theta) \pm P_{l-1}(\cos\theta)\,.
\end{align}
Furthermore, from the recursion formula for Legendre polynomials, it
follows that%
\footnote{This follows easily by induction from writing the recursion
  formula in the form
  $(l+1)(P_{l+1}\pm P_l) = \mp l(P_l\pm P_{l-1}) \pm ( 1 \pm
  \cos\theta) (2l+1)P_l$ and noting that
  $P_1 \pm P_0 = \cos\theta \pm 1$.}
\begin{equation}
  r^l[(P_l(\cos\theta) \pm P_{l-1}(\cos\theta)] = r (1\pm \cos\theta)
  \tilde G_{\pm} = \frac{X_\pm^2}{2} \tilde G_{\pm}
\end{equation}
for some analytic $\tilde{G}_{\pm}$, so indeed \eqref{cond_final} is
satisfied. This establishes that the 1-form $\omega$ is smooth, in fact analytic, at any centre corresponding to a corner of the orbit space.

Putting things together, we have shown that the spacetime metric is analytic at any point corresponding to a corner of the orbit space. Furthermore, near such points the spacetime is diffeomorphic to $\mathbb{R}^{1,4}$.

The gauge field in the new coordinates takes the form
\begin{equation}
  A = \frac{\sqrt{3}}{2} \big(  f\dt + A_+\dph^+ + A_-\dph^-\big)\,,
\end{equation}
where
\begin{equation}
  \label{gaugefield_centre}
  A_\pm = f \Big( \frac{2\omega_\psi X_\pm^2}{R^2} \pm h \omh_\phi \Big) \pm
  h \big(f\omega_\psi -\frac{K}{H}\big)\tilde{\chi}_\phi -
  \frac{2X_\pm^2}{R^2}\frac{K}{H} \mp h \xi_\phi\,.
\end{equation}
Clearly $A_t$ is analytic at $R=0$. We have already shown \eqref{cond_final}, so the first term in
\eqref{gaugefield_centre} is analytic and proportional to
$X_\pm^2$. As $f$, $\omega_\psi$, $K/H$ are analytic at
the centre and $\tilde{\chi}$ is of the form \eqref{chitilde}, the same
is true for the second term. Lastly, integrating \eqref{dxi}, for $\xi$ of the form (\ref{1forms}), gives
\begin{equation}
  \xi = \left( \xi_0 - k\cos\theta - \frac{X_+^2X_-^2}{4}\sum_{l=1}^{\infty}
  \frac{k_lr^{l-1}}{l+1}P'_l(\cos\theta) \right) \td \phi\,,
\end{equation}
and hence, using \eqref{1/H_centre},
\begin{equation}
  -\frac{2X_\pm^2}{R^2}\frac{K}{H}\mp h \xi_\phi = 
 - hk \mp h{\xi}_0 +
 X_\pm^2 (\ldots) + X_+^2X_-^2(\ldots)\,,
\end{equation}
where $\dots$ are analytic functions of $(X_+^2,X_-^2)$. Thus $A$ is
gauge-equivalent to an analytic 1-form. Therefore the Maxwell field (hence the
full solution) is analytic at the centre.

Finally, we emphasise that the above analysis shows that the solution is smooth and stably causal at and near any centre corresponding to a corner of the orbit space. Indeed we have,
\be
K^2+ HL =  \frac{1}{|f| r} + O(1), \qquad g^{tt}= -\frac{1}{f^2} + O(r)  \; ,
\ee
where recall that $f\neq 0$ at the centre, thus confirming the solution is smooth (\ref{smooth1}) and causal (\ref{causality}) near the centre.

\subsection{Summary}

To summarise, we have shown the following.

\begin{theorem} \label{thm:corner} Let $(M,g,F)$ be an asymptotically
  flat, supersymmetric and biaxisymmetric solution to minimal
  supergravity with a globally hyperbolic domain of outer
  communication $ \llangle M \rrangle$. Let $(\rho,z)=(0,z_i)$ be a point corresponding to a
  corner of the orbit space. Then the solution is
  smooth (in fact analytic) at the corner if and only if
  $f_i \equiv f|_{(\rho,z)=(0,z_i)}\neq 0$,
  $\omega_\psi |_{(\rho,z)=(0,z_i)} = 0$, $\chi|_{I}$ is an odd integer and $\omh|_{I} = 0$ on the adjacent axis rods $I$,
  and the harmonic functions $H$, $K$, $L$, $M$ are given by
  \begin{equation}
    H = \frac{h_i}{r_i} + \Ht_i\,,\qquad
    K = \frac{k_i}{r_i} + \Kt_i \,,\qquad
    L =- \frac{h_i^{-1} k_i^2}{r_i} + \Lt_i \,,\qquad
    M = \frac{\tfrac{1}{2} k_i^3}{r_i} + \Mt_i\,,
  \end{equation}
  where $r_i= \sqrt{ \rho^2+ (z-z_i)^2}$, $h_i= \text{sgn}(f_i) $,
  $k_i$ are constants and $\Ht_i$, $\Kt_i$, $\Lt_i$, $\Mt_i$ are
  harmonic functions on $\mathbb{R}^3$ which are smooth at
  $(\rho,z)=(0,z_i)$.  Furthemore, the spacetime near such a corner is diffeomorphic to $\mathbb{R}^{1,4}$.
\end{theorem}

\section{Moduli space of soliton and black hole solutions}
\label{sec:modspace}

\subsection{Classification theorem}

We will now combine the constraints obtained from the existence of a smooth horizon in section \ref{sec:hor} and smooth axes in section \ref{sec:axes} and give our main classification theorem.

\begin{theorem}\label{classthm} Consider an asymptotically flat,
  supersymmetric and biaxisymmetric solution to minimal supergravity
  with a smooth globally hyperbolic domain of outer communication and
  a smooth event horizon with compact cross-sections (if there is a
  black hole). Suppose the orbit space $\hat{M}$ has $k$ corners and
  the horizon has $l$ connected components
  ($l=0$ corresponds to no black hole), and let $n=k+l$. Then, the
  harmonic functions are
  \begin{equation}\label{harmonicfunc_classthm}
    H = \sum_{i=1}^{n} \frac{h_i}{r_i}\,,\qquad
    K = \sum_{i=1}^{n} \frac{k_i}{r_i}\,,\qquad
    L = 1 + \sum_{i=1}^{n} \frac{\ell_i}{r_i}\,,\qquad
    M = m + \sum_{i=1}^{n} \frac{m_i}{r_i}\,,
  \end{equation}
  where $r_i=\sqrt{\rho^2 + (z-z_i)^2}$, $(\rho,z)=(0,z_i)$ are
  corners of the orbit space or horizons, and
  \begin{equation}
    \sum_{i=1}^{n} h_i=1\,, \qquad m = -\tfrac{3}{2} \sum_{i=1}^n
    k_i\,.   \label{asympt}
  \end{equation}
  The corresponding 1-forms can be written as
  \begin{equation}
    \begin{aligned}\label{1forms_classthm}
      \chi &= \sum_{i=1}^n \frac{ h_i (z-z_i)}{r_i}  \td \phi \,, \qquad
      \xi = - \sum_{i=1}^n \frac{ k_i (z-z_i)}{r_i}  \td \phi \,,  \\
      \omh &= \left[- \sum_{i=1}^n \frac{(m h_i + \tfrac{3}{2}
          k_i) (z-z_i)}{r_i} - \sum_{i=1}^n \sum_{j \neq i} \left( \frac{h_i m_j+
            \tfrac{3}{2} k_i \ell_j}{z_i-z_j} \right)\left(
          \frac{\rho^2 + (z-z_i)(z-z_j)}{r_i r_j}-1 \right) \right]
      \td \phi\,,
    \end{aligned}
  \end{equation}
  and the parameters have to satisfy for each $i=1, \dots, n,$
  \begin{equation}
    h_i m + \tfrac{3}{2}k_i  
    + \sum_{\stackrel{j=1}{j\neq i}}^{n} \frac{h_im_j- m_i
      h_j-\frac{3}{2}(\ell_ik_j- k_i \ell_j) }{|z_i-z_j|}=
    0  \,. \label{conds2}
  \end{equation}
Furthermore, if $(0,z_i)$ is a corner, $h_i = \pm 1$  
and the parameters must satisfy
  \begin{align}
    & \ell_i= -h_i^{-1} k_i^2\,, \qquad m_i=
      \tfrac{1}{2}  k_i^3   \,,\label{conds1} \\
        &  h_i + \sum_{\stackrel{j=1}{j\neq i}}^{n}
      \frac{2k_ik_j-h_i(h_jk_i^2-\ell_j)}{|z_i-z_j|} >
      0 \,. \label{conds3}
  \end{align}
  On the other hand, if $(0,z_i)$ is a horizon the parameters must
  satisfy $h_i \in \mathbb{Z}$,
  \begin{equation}\label{conds_horizon}
    -h_i^2m_i^2 - 3h_ik_i\ell_im_i + h_i\ell_i^3 -
    2k_i^3m_i + \frac{3}{4} k_i^2\ell_i^2 > 0\, ,
  \end{equation}
 (which also implies $h_i\ell_i + k_i^2
>0$) and cross-sections of the horizon are of topology $S^3$ if
  $h_i=\pm 1$, $S^2\times S^1$ if $h_i=0$ and the lens space
  $L(h_i, 1)$ otherwise.
\end{theorem}

\begin{proof} We have shown that a horizon corresponds to at most a
  simple pole of the harmonic functions $H,K,L,M$, see Theorem \ref{thm:horizon}. Similarly, a corner
  corresponds to a simple pole of $H$ and at most a simple pole of
  $K,L,M$, see Theorem \ref{thm:corner}. Hence, with the stated assumptions we can write
  \begin{align}
    H =  \tilde{H}+ \sum_{i=1}^{n} \frac{h_i}{r_i}\,,\qquad
    K = \tilde{K}+ \sum_{i=1}^{n} \frac{k_i}{r_i}\,,\qquad
    L = \tilde{L}+  \sum_{i=1}^{n} \frac{\ell_i}{r_i}\,,\qquad
    M = \tilde{M} + \sum_{i=1}^{n} \frac{m_i}{r_i}\,,
  \end{align}
  where $\tilde{H}, \tilde{K}, \tilde{L}, \tilde{M}$ are harmonic
  functions smooth at $(\rho, z)=(0, z_i)$ for all $i=1, \dots, n$.
  By Lemma \ref{lemma:orbitspace}, the only singularities of $H, K, L, M$ in the DOC are at points corresponding to the corners of the
  orbit space, or at the horizon. Therefore,
  $\tilde{H}, \tilde{K}, \tilde{L}, \tilde{M}$ must be smooth on all
  of $\mathbb{R}^3$. Asymptotic flatness \eqref{Hasympt} implies these harmonic
  functions are bounded. Therefore,
  $\tilde{H}, \tilde{K}, \tilde{L}, \tilde{M}$ are smooth and bounded
  harmonic functions on $\mathbb{R}^3$.  Therefore, they must be
  constants which coincide with their asymptotic values, so
  $\tilde{H}= 0$, $\tilde{L}=1$, $\tilde{K}=0$ and $\tilde{M}=m$.
  The asymptotic flatness conditions \eqref{Hasympt} and
  \eqref{masympt} then reduce to \eqref{asympt}. This establishes the
  form of the harmonic functions. 
  
  Given the harmonic functions the
  1-forms are easily integrated using (\ref{dchi}, \ref{domegahat},
  \ref{dxi}). The integration constants in $\chi$ and $\omh$ have been fixed so
  that $\chi = \cos\theta +O(r^{-1})$ and $\omh = O(r^{-1})$ as $r \to \infty$, as required by asymptotic
  flatness.

  The constraints on the parameters at each corner \eqref{conds1} are
  given in Theorem \ref{thm:corner}. The additional constraint
  \eqref{conds3} is equivalent to the condition $h_i f_i >0$, which
  also follows from Theorem \ref{thm:corner}. The constraints on the
  parameters at a horizon \eqref{conds_horizon} are given in Theorems
  \ref{thm:horizon} and \ref{thm:horizon+AF}.
  
  The constraints
  \eqref{conds2} are equivalent to $\omh = 0$ on each of the axis rods $I_i=(z_i, z_{i+1})$,
  which is required by smoothness at the axes, see Lemma \ref{lemma:axis}. This can be seen as
  follows. From \eqref{1forms_classthm} it is obvious that $\omh$ is
  constant on any axis rod $I_i$. It can be shown that the difference between $\omh$ evaluated on two
  adjacent rods separated by the centre $(0,z_i)$ is  given
  by $-2$ times the lefthand side of \eqref{conds2}. Furthermore, by asymptotic
  flatness $\omh$ vanishes on the  rods $I_+=(z_n, \infty)$ and $I_-=(-\infty, z_1)$. We deduce that $\omh|_{I_i}=0$ for all $i=1, \dots n-1$
  precisely if \eqref{conds2} is
  satisfied, as claimed. It is worth noting that for any corner $(0,z_i)$ the condition (\ref{conds2}) is in fact equivalent to $\omega_\psi|_{r_i=0}=0$, as is also required by Theorem \ref{thm:corner}.
  
  Finally,  we note that the other
  condition required for smooth axes, given in Lemma \ref{lemma:axis}, is that
  $\chi$ evaluated on each axis rod has to be an odd integer.  Evaluating  \eqref{1forms_classthm} on each axis rod, we find this is automatically satisfied since $h_i$ are integers for all $i=1, \dots, n$ (see equation \eqref{chiIi}).
  
  \end{proof}
\vspace{1.0cm}
\noindent {\bf Remarks}
\begin{enumerate}
\item This shows that supersymmetric black holes and
solitons, {\it must} be multi-centred solutions with a Gibbons--Hawking base. This is a five-dimensional analogue of Corollary 4.2
in~\cite{Chrusciel:2005ve}.  
\item To confirm that the solution is smooth and stably causal everywhere in the DOC one must check the condition (\ref{smooth1}) in Lemma \ref{GHlemma} and \eqref{causality}. Our analysis shows that these are indeed satisfied near infinity, near the horizon and near any point corresponding to a corner of the orbit space. We have been unable to check that the conditions listed in Theorem \ref{classthm} are sufficient to ensure smoothness and causality are obeyed everywhere else in the DOC. However, based on the known examples (discussed below) we believe that no further conditions on the parameters arise. Nevertheless, it is possible that (\ref{smooth1}) and (\ref{causality}) may impose
additional constraints on the parameters of the solution.
\end{enumerate}

We can see from Theorem \ref{classthm} that a general solution
with $n = k + l$ centres will be determined by the (discrete)
$n$-dimensional vector $h = (h_1,\ldots,h_n)$, as well as $(4n -1)$
real parameters,
\begin{equation}
\Big\{ \{z_{i+1}-z_i\}_{i=1,\ldots,n-1},\,\{k_i,\ell_i,m_i\}_{i=1,\ldots,n}\Big\}\,,
\end{equation}
subject to $3k+l$ constraint equations
\eqref{conds2}--\eqref{conds1}, of which \eqref{conds1} can be solved
algebraically. Furthermore, there is a remaining one-parameter gauge freedom
\eqref{Kshift}--\eqref{LMshift} in the harmonic functions under which
\begin{equation}
  \label{shift}
  k_i \to k_i + ch_i\,,\qquad
  \ell_i \to \ell_i-2ck_i-2c^2h_i\,,\qquad
  m_i \to m_i -\tfrac{3}{2}c\ell_i + \tfrac{3}{2}c^2k_i +
  \tfrac{1}{2}c^3h_i\,.
\end{equation}
Summing up, we find that the moduli space of $(k+l)$-centred solutions,
$\mathcal{M}^{k,l}$, is given by the subset  of the $(2k+4l-1)$-dimensional
parameter space
\begin{equation}
\Big\{ \{z_{i+1}-z_i\}_{i=1,\ldots,n-1},\,\{k_i\}_{i=1,\ldots,n},\,
\{\ell_j,m_j\}_{\text{if $z_j$ is a horizon}}\Big\}\,,
\end{equation}
defined by the set of $k+l$
polynomial equations \eqref{conds2} subject to the inequalities \eqref{conds3} and equivalence relations
\eqref{shift}.  By a general count of degrees of
freedom,
\begin{equation}
  \dim \mathcal{M}^{k,l} = k + 3l - 2 - \tilde{\Delta} + \Delta(k,l) \,,
\end{equation}
where $\tilde{\Delta}$ has been introduced to correct for any
potential restrictions on the parameters coming from (\ref{smooth1}, \ref{causality}) (see Remark 2 above),
and the second correction term $\Delta(k,l)$ to accomodate for a potential
redundancy in equations \eqref{conds2}. One can easily see that there
is at least one
such redundancy as summing 
\eqref{conds2} over all $i$ gives
\begin{equation}
  m\sum_{i=1}^{n}h_i +\frac{3}{2}\sum_{i=1}^n k_i + \sum_{i=1}^n
  \sum_{\stackrel{j=1}{j\neq i}}^{n}
   \frac{h_im_j- m_i
    h_j-\frac{3}{2}(\ell_ik_j- k_i \ell_j) }{|z_i-z_j|} = 0\,
\end{equation}
where we have made use of \eqref{asympt} and \eqref{conds1}, and the
double sum vanishes for reasons of symmetry.
We thus know that 
\begin{equation}
  1 \leq \Delta(k,l)  \leq k+l.
\end{equation}
In fact on the basis of known examples, we will conjecture that
$\tilde{\Delta} = 0$, $\Delta(k,l) = 1$, so
\begin{equation}
  \dim \mathcal{M}^{k,l} = k + 3l - 1 \,.   \label{dimmoduli}
\end{equation}
Indeed, this agrees with the known solutions which are discussed below.

When counting the number of solutions it is important to realise there is a redundancy in our parameterisation corresponding to a discrete global isometry,
\begin{equation}\label{reflection}
  z\to -z\,,\qquad
  z_i \to -z_{n-i+1}\,,\qquad
  \phi \to - \phi\,,\qquad 
  i \to n-i+1\,.
\end{equation}
Now, each separate choice of $h$ will define a component of the moduli
space. The number of connected components of $\mathcal{M}^{k,l}$ is
thus given by the number of possible choices of $h$, taking into
account the remaining reflection symmetry \eqref{reflection} of
the axis. As we will show next, the choice of $h$ is precisely equivalent to the rod structure of
the solution, so the number of components of the moduli space is also the number of inequivalent rod structures. 

As we have seen earlier, the centres $z=z_i$ split the $z$-axis into
$n+1$ intervals, $I_\pm, I_i$, on each of which the respective
Killing field \eqref{rodvecs_euler} vanishes. Having the full solution
at hand, we can now explicitly evaluate $\chi$ on each of
these intervals as
\begin{equation} \label{chiIi}
   \chi_\pm = \pm 1, \qquad
  \chi_i \equiv \chi|_{I_i} = \sum_{j=1}^i h_j - \sum_{j=i+1}^n
    h_j
  =  2\sum_{j=1}^i h_j - 1 \, ,
\end{equation}
where the final equality follows from the asymptotic condition (\ref{asympt}).
Therefore, in the basis defined by \eqref{vpm}, one finds the rod
vectors are given by \eqref{rodvecs} with
\begin{equation}
a_i = \sum_{j=1}^{i}h_j\,.
\end{equation}
The determinants of adjacent rod vectors \eqref{detrodvecs} are then
precisely given by the value of $h_i$ at the respective centre,
\begin{equation}
  \det ( v^T_- \; v^T_1)  =h_1\,, \qquad
  \det ( v^T_i  \; v^T_{i+1} ) = h_{i+1}\,,\qquad
  \det ( v^T_n \;  v^T_+) = h_n\,.
\end{equation}
Therefore, our horizon and axes analysis, which showed that $h_i=p$ for a centre corresponding to an $L(p,q)$ horizon, and $h_i=\pm 1$ for a centre corresponding to a corner of the orbit space, precisely agree with the compatibility conditions for adjacent rod vectors previously derived for stationary and biaxisymmetric spacetimes~\cite{Hollands:2007aj}. Therefore, these compatibility conditions impose no extra constraints.

The topology of the domain of outer communication is nontrivial and determined by the rod structure.
The internal axis rods
$I_i$ ($i = 1,\ldots,n-1$), or indeed any simple curve in the $\mathbb{R}^3$ GH base between the endpoints of $I_i$, together with the $U(1)$ $\psi$-fibre over the GH base, correspond to noncontractible 2-cycles $C_i$.  If the endpoints of $I_i$ are both corners of the orbit space the $\psi$-fibre collapses smoothly at the endpoints, so $C_i$ is a surface of $S^2$ topology.  If one endpoint of $I_i$ is a corner and one a horizon, then $C_i$ is a surface of 2-disc topology, with the boundary of the disc attached to the horizon. Finally, if both endpoints are horizons then $C_i$ is a 2-tube with each of its boundaries attached to one horizon.

\subsection{Soliton solutions}

Let us first consider the moduli space of $n$-centred soliton
solutions, $\mathcal{M}^{n,0}$. Since every centre corresponds to a corner of the orbit space we must have $h_i = \pm 1$ for all
$i = 1,\ldots n$. On the other hand, asymptotic flatness requires $\sum_{i=1}^n h_i=1$. It follows that soliton solutions will necessarily have an odd
number of centres, $n = 2m + 1$, where $m$ is the number of $h_i=-1$.  We can now easily determine the number of distinct rod structures this
allows for. There are $\binom{n}{m}$ possible ways of
choosing $h\equiv(h_1,\ldots,h_n)$. Some of these, however, will be
related by the discrete reflection symmetry (\ref{reflection}) (which implies
$h_i\to h_{n-i+1}$) and thus correspond to isometric
solutions. Correcting for this overcounting, one finds the number of
connected components of the moduli space to be given by
\begin{equation}
N(\mathcal{M}^{n,0})=\frac{1}{2}\left[\binom{n}{m} +
  \binom{m}{[m/2]}\right]\,,
\end{equation}
where the latter term arises as a correction for solutions which are
themselves symmetric under reflection (and thus had \emph{not} falsely
been overcounted before). 

For $n=1$, the only possible
solution without a black hole is Minkowski space. The allowed inequivalent rod structures for
$n=3$ are defined by $h =(1,1,-1)$ and $h=(1,-1,1)$ and are depicted
in figure \ref{fig:solitons}.  
\begin{figure}[h!]
\centering
%%%%% 1st picture:
\subfloat[$h=(1,1,-1)$]{
\begin{tikzpicture}[scale=1, every node/.style={scale=0.7}]
  \draw[very thick](-4,0)--(-1.9,0)node[midway,above=.2cm]{$(1,0)$};
  \draw[very thick](-1.7,0)--(0.0,0)node[midway,above=0.2cm]{$(0,1)$};
  \draw[very thick](0.2,0)--(2.0,0)node[midway,above= 0.2cm]{$(-1,2)$};
  \draw[very thick](2.2,0)--(4.2,0)node[midway,above= 0.2cm]{$(0,1)$};
  \draw[fill=black] (-1.8,0) circle [radius=.1];
  \draw[fill=black] (0.1,0) circle [radius=.1]; 
  \draw[fill=black] (2.1,0) circle [radius=.1];
\end{tikzpicture}}
\hspace{2em}
%%%%% 2nd picture:
\subfloat[$h=(1,-1,1)$]{
\begin{tikzpicture}[scale=1, every node/.style={scale=0.7}]
  \draw[very thick](-4,0)--(-1.9,0)node[midway,above=.2cm]{$(1,0)$};
  \draw[very thick](-1.7,0)--(0.0,0)node[midway,above=.2cm]{$(0,1)$};
  \draw[very thick](0.2,0)--(2.0,0)node[midway,above=.2cm]{$(1,0)$};
  \draw[very thick](2.2,0)--(4.2,0)node[midway,above=.2cm]{$(0,1)$};
  \draw[fill=black] (-1.8,0) circle [radius=.1];
  \draw[fill=black] (0.1,0) circle [radius=.1];
  \draw[fill=black] (2.1,0) circle [radius=.1];
\end{tikzpicture}}
\caption{Inequivalent rod structures for 3-centred solitons.}
\label{fig:solitons}
\end{figure}
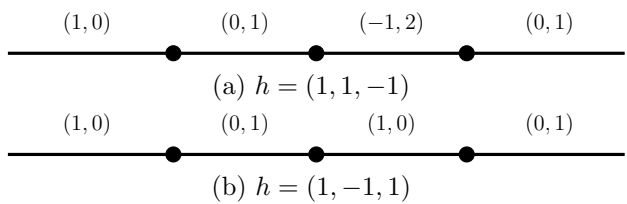
In particular, we see that there are two inequivalent soliton solutions in this case. The above counting formula shows that the number of inequivalent soliton solutions increases with $n$.

The $n$-centred soliton solutions correspond to asymptotically flat, globally hyperbolic regular spacetimes containing $n-1$ noncontracible 2-cycles, or `bubbles'.  Such bubbling spacetimes were first constructed in~\cite{Bena:2005va} and some their global properties elucidated in~\cite{Gibbons:2013tqa}.

\subsection{Single black hole solutions}

We now consider the moduli space for $n$-centred solutions with a single black hole, $\mathcal{M}^{n-1,1}$.  Thus, for one centre, say $z_j$,
the determinant of the matrix of adjacent rod vectors
$h_j = p \in \mathbb{Z}$ while the other centres correspond to corners, so 
$h_i = \pm 1$ for all $i\neq j$. As we have seen, this means that the
centre $z=z_j$ corresponds to a horizon of topology
$L(p,1)$. Denote the number of corners with $h_i=\pm 1$ by $n_\pm$ so $n_++n_-+1=n$. Asymptotic flatness \eqref{asympt} also implies that $n_+-n_-+p=1$. It follows that 
\be
p= n - 2 n_+
\ee
where $0\leq n_+ \leq n-1$. Hence $p$  is even for an even number of centres and odd otherwise and the possible values of $p$ are $-n+2, -n+4, \dots, n-2, n$.  

For a given $p$, there are $n \binom{n-1}{n_+}$ ways of choosing $h$. However, some of these configurations will be related by the reflection symmetry (\ref{reflection}) and hence they are double counted. To determine this number, we first must identify the number of configurations which are symmetric under the reflection. Symmetric rod structures can only occur for odd $n$ and even $n_+$, with the middle centre corresponding to the horizon, in which case there are $\binom{(n-1)/2}{n_+/2}$ such symmetric configurations. Putting all this together, we find that the number of components of the moduli space of single black hole solutions with $L(p,1)$ topology, $\mathcal{M}^{n-1,1}_{p}$, is
\begin{equation}\label{singleLp}
N(\mathcal{M}_p^{n-1,1}) = 
\begin{cases}
\frac{n}{2} \binom{n-1}{n_+}+ \frac{1}{2}\binom{(n-1)/2}{n_+/2} & \text{if $n$ odd and $n_+$ even}\\
\frac{n}{2} \binom{n-1}{n_+} & \text{otherwise}
\end{cases}
\end{equation}
Summing over the possible $p$ we find that the total number is
\begin{equation}\label{singleBH}
  N(\mathcal{M}^{n-1,1}) = \sum_{n_+=0}^{n-1}N(\mathcal{M}^{n-1,1}_{p}) = 
  \begin{cases}
  n 2^{n-2} & \text{if $n$ is even} \\
  n 2^{n-2} + 2^{\frac{n-3}{2}}  & \text{if $n$ is odd} 
  \end{cases}
\end{equation}
Let us consider a few examples. 

The simplest possibility is $n=1$, which implies $n_+=0 $ and $p=1$ and hence the horizon topology is $S^3$. This of course corresponds to the BMPV black hole~\cite{Breckenridge:1996is}.

Now let us consider the $n=2$ case. From \eqref{singleBH} we find that there are 2 classes of
two-centred single black hole solutions, whose rod structures are
shown in figure \ref{fig:2centres}. The first of these, figure
\ref{2lens} is the recently constructed $L(2,1)$ black lens~\cite{Kunduri:2014kja}. Figure \ref{ring}
corresponds to the known supersymmetric black ring solution~\cite{Elvang:2004rt}.
\begin{figure}[h!]
\centering
\subfloat[$p=2$, $h=(2,-1)$]{
\begin{tikzpicture}[scale=1, every node/.style={scale=0.7}]
  \draw[very thick](-3,0)--(-1.1,0)node[midway,above=.2cm]{$(1,0)$};
  \draw[very thick](-0.9,0)--(0.9,0)node[midway,above=0.2cm]{$(-1,2)$}; 
  \draw[very thick](1.1,0)--(3.0,0)node[midway, above= 0.2cm]{$(0,1)$};
  \draw[fill=white] (-1,0) circle [radius=.1] node[above=.2cm]{$H$};
  \draw[fill=black] (1,0) circle [radius=.1]; 
\end{tikzpicture}\label{2lens}}
\hspace{2em}
\subfloat[$p=0$, $h=(0,1)$.]{
\begin{tikzpicture}[scale=1, every node/.style={scale=0.7}]
  \draw[very thick](-3,0)--(-1.1,0)node[midway,above=.2cm]{$(1,0)$};
  \draw[very thick](-0.9,0)--(0.9,0)node[midway,above=0.2cm]{$(1,0)$}; 
  \draw[very thick](1.1,0)--(3.0,0)node[midway, above= 0.2cm]{$(0,1)$};
  \draw[fill=white] (-1,0) circle [radius=.1] node[above=.2cm]{$H$};
  \draw[fill=black] (1,0) circle [radius=.1]; 
\end{tikzpicture}\label{ring}}
\caption{Rod structures for 2-centred single black hole solutions.}
\label{fig:2centres}
\end{figure}
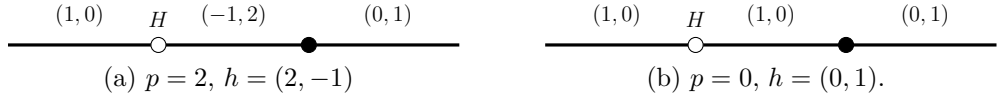

Next, for $n=3$ we see there are seven distinct rod structures. These are depicted in figures
\ref{fig:3centres1} and \ref{fig:3centres2}. There are two inequivalent black holes with a horizon of topology
$L(3,1)$, of which only figure \ref{3lens1} corresponds to the solution
constructed in \cite{Tomizawa:2016kjh}. There are five inequivalent $S^3$ black holes, of which only figure \ref{bubble}
corresponds to the known $S^3$ black hole with bubble~\cite{Kunduri:2014iga}. The other solutions had not previously been constructed.

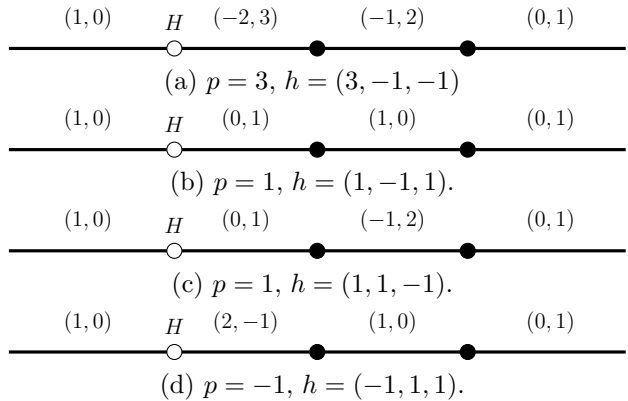
\begin{figure}[h!]
\centering
\subfloat[$p=3$, $h=(3,-1,-1)$]{
\begin{tikzpicture}[scale=1, every node/.style={scale=0.7}]
  \draw[very thick](-4,0)--(-1.9,0)node[midway,above=.2cm]{$(1,0)$};
  \draw[very thick](-1.7,0)--(0.0,0)node[midway,above=0.2cm]{$(-2,3)$}; 
  \draw[very thick](0.2,0)--(2.0,0)node[midway, above= 0.2cm]{$(-1,2)$};
  \draw[very thick](2.2,0)--(4.2,0)node[midway, above= 0.2cm]{$(0,1)$};
  \draw[fill=white] (-1.8,0) circle [radius=.1] node[above=.2cm]{$H$};
  \draw[fill=black] (0.1,0) circle [radius=.1]; 
  \draw[fill=black] (2.1,0) circle [radius=.1];
\end{tikzpicture}\label{3lens1}}
\hspace{2em}
\subfloat[$p=1$, $h=(1,-1,1)$.]{
\begin{tikzpicture}[scale=1, every node/.style={scale=0.7}]
  \draw[very thick](-4,0)--(-1.9,0)node[midway,above=.2cm]{$(1,0)$};
  \draw[very thick](-1.7,0)--(0.0,0)node[midway,above=.2cm]{$(0,1)$};
  \draw[very thick](0.2,0)--(2.0,0)node[midway,above=.2cm]{$(1,0)$};
  \draw[very thick](2.2,0)--(4.2,0)node[midway,above=.2cm]{$(0,1)$};
  \draw[fill=white] (-1.8,0) circle [radius=.1] node[above=.2cm]{$H$};
  \draw[fill=black] (0.1,0) circle [radius=.1];
  \draw[fill=black] (2.1,0) circle [radius=.1];
\end{tikzpicture}\label{bubble}}
\hspace{2em}
\subfloat[$p=1$, $h=(1,1,-1)$.]{
\begin{tikzpicture}[scale=1, every node/.style={scale=0.7}]
  \draw[very thick](-4,0)--(-1.9,0)node[midway,above=.2cm]{$(1,0)$};
  \draw[very thick](-1.7,0)--(0.0,0)node[midway,above=.2cm]{$(0,1)$};
  \draw[very thick](0.2,0)--(2.0,0)node[midway,above=.2cm]{$(-1,2)$};
  \draw[very thick](2.2,0)--(4.2,0)node[midway,above=.2cm]{$(0,1)$};
  \draw[fill=white] (-1.8,0) circle [radius=.1] node[above=.2cm]{$H$};
  \draw[fill=black] (0.1,0) circle [radius=.1];
  \draw[fill=black] (2.1,0) circle [radius=.1];
\end{tikzpicture}}
\hspace{2em}
\subfloat[$p=-1$, $h=(-1,1,1)$.]{
\begin{tikzpicture}[scale=1, every node/.style={scale=0.7}]
  \draw[very thick](-4,0)--(-1.9,0)node[midway,above=.2cm]{$(1,0)$};
  \draw[very thick](-1.7,0)--(0.0,0)node[midway,above=.2cm]{$(2,-1)$};
  \draw[very thick](0.2,0)--(2.0,0)node[midway,above=.2cm]{$(1,0)$};
  \draw[very thick](2.2,0)--(4.2,0)node[midway,above=.2cm]{$(0,1)$};
  \draw[fill=white] (-1.8,0) circle [radius=.1] node[above=.2cm]{$H$};
  \draw[fill=black] (0.1,0) circle [radius=.1];
  \draw[fill=black] (2.1,0) circle [radius=.1];
\end{tikzpicture}}
\caption{Rod structures for 3-centred single black holes
  with the horizon at the first centre.}
\label{fig:3centres1}
\end{figure}

\begin{figure}[h!]
\centering
\subfloat[$p=3$, $h=(-1,3,-1)$]{
  \begin{tikzpicture}[scale=1, every node/.style={scale=0.7}]
    \draw[very thick](-4,0)--(-1.9,0)node[midway,above=.2cm]{$(1,0)$};
    \draw[very thick](-1.7,0)--(0.0,0)node[midway,above=.2cm]{$(2,-1)$};
    \draw[very thick](0.2,0)--(2.0,0)node[midway,above=.2cm]{$(-1,2)$};
    \draw[very thick](2.2,0)--(4.2,0)node[midway,above=.2cm]{$(0,1)$};
    \draw[fill=black] (-1.8,0) circle [radius=.1] ;
    \draw[fill=white] (0.1,0) circle [radius=.1] node[above=.2cm]{$H$};
    \draw[fill=black] (2.1,0) circle [radius=.1];
\end{tikzpicture}}
\hspace{2em}
\subfloat[$p=1$, $h=(-1,1,1)$.]{
\begin{tikzpicture}[scale=1, every node/.style={scale=0.7}]
  \draw[very thick](-4,0)--(-1.9,0)node[midway,above=.2cm]{$(1,0)$};
  \draw[very thick](-1.7,0)--(0.0,0)node[midway,above=.2cm]{$(2,-1)$};
  \draw[very thick](0.2,0)--(2.0,0)node[midway,above=.2cm]{$(1,0)$};
  \draw[very thick](2.2,0)--(4.2,0)node[midway,above=.2cm]{$(0,1)$};
  \draw[fill=black] (-1.8,0) circle [radius=.1];
  \draw[fill=white] (0.1,0) circle [radius=.1] node[above=.2cm]{$H$};
  \draw[fill=black] (2.1,0) circle [radius=.1];
\end{tikzpicture}}
\hspace{2em}
\subfloat[$p=-1$, $h=(1,-1,1)$.]{
\begin{tikzpicture}[scale=1, every node/.style={scale=0.7}]
  \draw[very thick](-4,0)--(-1.9,0)node[midway,above=.2cm]{$(1,0)$};
  \draw[very thick](-1.7,0)--(0.0,0)node[midway,above=.2cm]{$(0,1)$};
  \draw[very thick](0.2,0)--(2.0,0)node[midway,above=.2cm]{$(1,0)$};
  \draw[very thick](2.2,0)--(4.2,0)node[midway,above=.2cm]{$(0,1)$};
  \draw[fill=black] (-1.8,0) circle [radius=.1];
  \draw[fill=white] (0.1,0) circle [radius=.1] node[above=.2cm]{$H$};
  \draw[fill=black] (2.1,0) circle [radius=.1] ;
\end{tikzpicture}}
\caption{Rod structures for 3-centred single black holes
  with a central horizon.}
\label{fig:3centres2}
\end{figure}
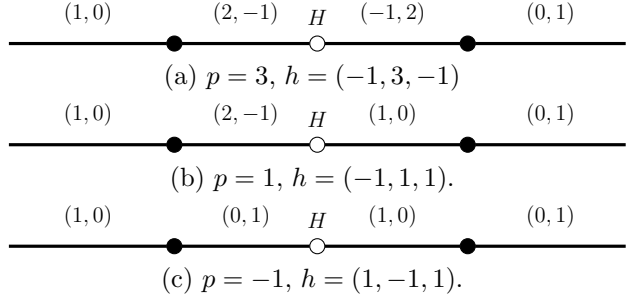

 More generally, we see that a single $S^3$ black hole, so $p=\pm 1$, requires an odd number of centres. Such solutions correspond to a spherical black hole in a bubbling spacetime with  $n-2$ bubbles (and $1$ disc), or $n-3$ bubbles (and $2$ discs), depending on which centre corresponds to the horizon, and have not been previously constructed.
 
Now consider single black hole solutions with $S^1\times S^2$ horizon topology,  so $p=0$. These must have an even number of centres $n$ and from the above we see that there are $\frac{n}{2} \binom{n-1}{n/2}$ inequivalent $n$-centred solutions with a single black ring. For even $n>2$ we find there are an increasing number of inequivalent black ring in bubbling spacetime solutions which have not previously been discussed. For example, in figure \ref{fig:4centres1}, we list the six possible rod structures for 4-centred single black ring solutions. 

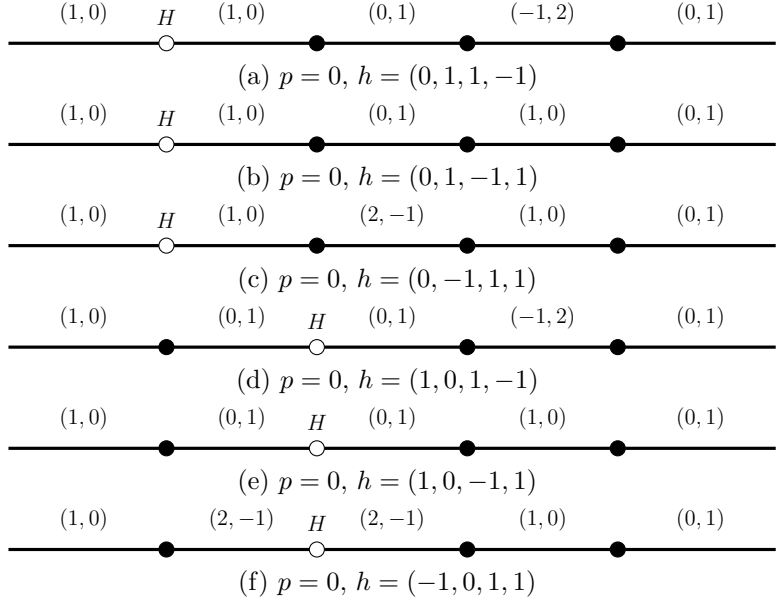
\begin{figure}[h!]
\centering
\subfloat[$p=0$, $h=(0,1,1,-1)$]{
  \begin{tikzpicture}[scale=1, every node/.style={scale=0.7}]
    \draw[very thick](-5.1,0)--(-3.1,0)node[midway,above=.2cm]{$(1,0)$};
    \draw[very thick](-2.9,0)--(-1.1,0)node[midway,above=.2cm]{$(1,0)$};
    \draw[very thick](-0.9,0)--(0.9,0)node[midway,above=.2cm]{$(0,1)$};
    \draw[very thick](1.1,0)--(2.9,0)node[midway,above=.2cm]{$(-1,2)$};
    \draw[very thick](3.1,0)--(5.1,0)node[midway,above=.2cm]{$(0,1)$};
    \draw[fill=white] (-3,0) circle [radius=.1] node[above=.2cm]{$H$};
    \draw[fill=black] (-1,0) circle [radius=.1] ;
    \draw[fill=black] (1,0) circle [radius=.1] ;
    \draw[fill=black] (3,0) circle [radius=.1];
\end{tikzpicture}}
\hspace{2em}
\subfloat[$p=0$, $h=(0,1,-1,1)$]{
  \begin{tikzpicture}[scale=1, every node/.style={scale=0.7}]
    \draw[very thick](-5.1,0)--(-3.1,0)node[midway,above=.2cm]{$(1,0)$};
    \draw[very thick](-2.9,0)--(-1.1,0)node[midway,above=.2cm]{$(1,0)$};
    \draw[very thick](-0.9,0)--(0.9,0)node[midway,above=.2cm]{$(0,1)$};
    \draw[very thick](1.1,0)--(2.9,0)node[midway,above=.2cm]{$(1,0)$};
    \draw[very thick](3.1,0)--(5.1,0)node[midway,above=.2cm]{$(0,1)$};
    \draw[fill=white] (-3,0) circle [radius=.1] node[above=.2cm]{$H$};
    \draw[fill=black] (-1,0) circle [radius=.1] ;
    \draw[fill=black] (1,0) circle [radius=.1] ;
    \draw[fill=black] (3,0) circle [radius=.1];
\end{tikzpicture}}
\hspace{2em}
\subfloat[$p=0$, $h=(0,-1,1,1)$]{
  \begin{tikzpicture}[scale=1, every node/.style={scale=0.7}]
    \draw[very thick](-5.1,0)--(-3.1,0)node[midway,above=.2cm]{$(1,0)$};
    \draw[very thick](-2.9,0)--(-1.1,0)node[midway,above=.2cm]{$(1,0)$};
    \draw[very thick](-0.9,0)--(0.9,0)node[midway,above=.2cm]{$(2,-1)$};
    \draw[very thick](1.1,0)--(2.9,0)node[midway,above=.2cm]{$(1,0)$};
    \draw[very thick](3.1,0)--(5.1,0)node[midway,above=.2cm]{$(0,1)$};
    \draw[fill=white] (-3,0) circle [radius=.1] node[above=.2cm]{$H$};
    \draw[fill=black] (-1,0) circle [radius=.1] ;
    \draw[fill=black] (1,0) circle [radius=.1] ;
    \draw[fill=black] (3,0) circle [radius=.1];
\end{tikzpicture}}
\hspace{2em}
\subfloat[$p=0$, $h=(1,0,1,-1)$]{
 \begin{tikzpicture}[scale=1, every node/.style={scale=0.7}]
   \draw[very thick](-5.1,0)--(-3.1,0)node[midway,above=.2cm]{$(1,0)$};
   \draw[very thick](-2.9,0)--(-1.1,0)node[midway,above=.2cm]{$(0,1)$};
   \draw[very thick](-0.9,0)--(0.9,0)node[midway,above=.2cm]{$(0,1)$};
   \draw[very thick](1.1,0)--(2.9,0)node[midway,above=.2cm]{$(-1,2)$};
   \draw[very thick](3.1,0)--(5.1,0)node[midway,above=.2cm]{$(0,1)$};
   \draw[fill=black] (-3,0) circle [radius=.1];
   \draw[fill=white] (-1,0) circle [radius=.1] node[above=.2cm]{$H$} ;
   \draw[fill=black] (1,0) circle [radius=.1] ;
   \draw[fill=black] (3,0) circle [radius=.1];
\end{tikzpicture}}
\hspace{2em}
\subfloat[$p=0$, $h=(1,0,-1,1)$]{
 \begin{tikzpicture}[scale=1, every node/.style={scale=0.7}]
   \draw[very thick](-5.1,0)--(-3.1,0)node[midway,above=.2cm]{$(1,0)$};
   \draw[very thick](-2.9,0)--(-1.1,0)node[midway,above=.2cm]{$(0,1)$};
   \draw[very thick](-0.9,0)--(0.9,0)node[midway,above=.2cm]{$(0,1)$};
   \draw[very thick](1.1,0)--(2.9,0)node[midway,above=.2cm]{$(1,0)$};
   \draw[very thick](3.1,0)--(5.1,0)node[midway,above=.2cm]{$(0,1)$};
   \draw[fill=black] (-3,0) circle [radius=.1];
   \draw[fill=white] (-1,0) circle [radius=.1] node[above=.2cm]{$H$} ;
   \draw[fill=black] (1,0) circle [radius=.1] ;
   \draw[fill=black] (3,0) circle [radius=.1];
\end{tikzpicture}}
\hspace{2em}
\subfloat[$p=0$, $h=(-1,0,1,1)$]{
 \begin{tikzpicture}[scale=1, every node/.style={scale=0.7}]
   \draw[very thick](-5.1,0)--(-3.1,0)node[midway,above=.2cm]{$(1,0)$};
   \draw[very thick](-2.9,0)--(-1.1,0)node[midway,above=.2cm]{$(2,-1)$};
   \draw[very thick](-0.9,0)--(0.9,0)node[midway,above=.2cm]{$(2,-1)$};
   \draw[very thick](1.1,0)--(2.9,0)node[midway,above=.2cm]{$(1,0)$};
   \draw[very thick](3.1,0)--(5.1,0)node[midway,above=.2cm]{$(0,1)$};
   \draw[fill=black] (-3,0) circle [radius=.1] ;
   \draw[fill=white] (-1,0) circle [radius=.1] node[above=.2cm]{$H$} ;
   \draw[fill=black] (1,0) circle [radius=.1] ;
   \draw[fill=black] (3,0) circle [radius=.1];
\end{tikzpicture}}
\caption{Rod structures for 4-centred single black ring solutions}
\label{fig:4centres1}
\end{figure}

\subsection{Multi black hole solutions}

We will not consider the case of multi black holes in
detail. Previously constructed examples in this class are the multi
black rings~\cite{Gauntlett:2004wh}, a double $S^3$ black
hole~\cite{Crichigno:2016lac} and more generally multi black
lenses~\cite{Tomizawa:2017suc}.  We emphasise that the multi extreme
Reissner-Nordstrom and multi BMPV black hole solutions
\cite{Gauntlett:1998fz} do not fit into our classification as they are not biaxisymmetric (they preserve at most $SO(3)$ rotational symmetry).

\subsection{Physical properties}

The mass and angular momenta for the general solution described in
Theorem \ref{classthm} are given by (using \eqref{asympt})
\begin{align}
  M &= 3\pi\Big(\sum_{i=1}^{n}\ell_i + \frac{4}{9}m^2 \Big)\,,\label{mass}\\
  J_\psi &= 2\pi \sum_{i=1}^{n}\Big(\frac{4}{9}m^2k_i -m\ell_i + m_i
           \Big)\,,\label{Jpsi}\\
  J_\phi &= 2\pi \sum_{i=1}^{n}z_i\Big(mh_i + \frac{3}{2}k_i \Big)\,,\label{Jphi}
\end{align}
and the electric charge
\begin{equation}\label{charge}
   Q = \frac{1}{4\pi} \int_{S^3}\star F = 2\sqrt{3}
   \pi\Big(\sum_{i=1}^{n}\ell_i + \frac{4}{9}m^2 \Big)
\end{equation}
satisfies the BPS bound $Q = \frac{2}{\sqrt{3}} M$. 

As noted above, the general solution possesses nontrivial topology in the form of 2-cycles (bubbles, discs, tubes) corresponding to the finite axis rods $I_i$, where $i=1, \dots, n-1$. The fluxes through these noncontractible
2-cycles $C_i$ are given by
\begin{equation}
  \Pi[C_i] = \frac{1}{4\pi}\int_{C_i}F =\frac{\sqrt{3}}{2}
  \frac{h_jm_j+\frac{1}{2}k_j\ell_j}{h_j\ell_j+k_j^2}\Big|_{j=i}^{j=i+1}\,.
\end{equation}
Note that for a corner $z_j$, the expression on the right hand
side simplifies to
$(h_jm_j+\frac{1}{2}k_j\ell_j)/(h_k\ell_j+k_j^2) =
-k_j/h_j$.  The nontrivial topology also allows us to define constant magnetic potentials $\Phi_i$
associated with each axis rod $I_i$ by~\cite{Kunduri:2013vka}
\begin{equation}
  \iota_{v_i}F = \d \Phi_i
\end{equation}
where we fix $\Phi_i \to 0$ asymptotically. We find that, for $i=1, \dots, n-1$, the $\Phi_i$ evaluated on the corresponding axis rods are 
\begin{equation}\label{potentials}
q_i\equiv   \Phi_i|_{I_i} = \frac{\sqrt{3}}{2}\Big((\chiph|_{I_i}-1)\sum_{j=1}^i k_j +
   (\chiph|_{I_i}+1)\sum_{j=i+1}^nk_j\Big)\,,
\end{equation}
 which are indeed constants.

Thus a solution with $n$-centres carries the global charges
$Q,J_\psi, J_\phi$ (with $M$ fixed by $Q$) and also $n-1$ local
magnetic potentials $q_i$ (or  magnetic fluxes
$\Pi[C_i]$), leading to a total of $n+2$ physical charges. On the other
hand, the dimension of the moduli space (\ref{dimmoduli}) for a
solution with a single black hole is $n+1$ and for a soliton is
$n-1$. Therefore, for a single black hole there the must be a single
constraint on the $n+2$ physical parameters, whereas for a soliton
there must be three such constraints. 

The constraints on the physical parameters can be seen more explicitly
as follows. Using the constraints on the parameters  \eqref{conds2} one can show that
\begin{align}
  J_\phi = -2\pi \sum_{i=1}^{n} \sum_{j<i}
  \left(h_im_j-m_ih_j-\tfrac{3}{2}(\ell_ik_j-k_i\ell_j)\right)\,.\label{Jphi_new}
\end{align}
%As we can always set $k_j=0$ for at least one $j \in \{1,\cdots,n\}$
%by exploiting the gauge freedom \eqref{shift}, \eqref{potentials}
%defines a 1-1-correspondence between the remaining parameters $k_i$
%and the magnetic potentials $\Phi^i$. 
The gauge freedom  \eqref{shift} implies that we can always set $k_j=0$ for at least one $j \in \{1,\cdots,n\}$ (since at least one $h_i$ must be nonvanishing due to (\ref{asympt})). Therefore, we may invert  \eqref{potentials} to express the remaining $n-1$ parameters $k_{i\neq j}$ as linear combinations of the $n-1$ magnetic potentials $q_i$ (one can check the matrix relating the two sets of quantities is indeed invertible). This gives a direct physical interpretation to the parameters $k_i$.  At every corner $z_i$ the parameters $\ell_i, m_i$ are determined in terms of $k_i$ and hence can also be expressed solely in terms of the magnetic potentials. In the case of
a single black hole at position $z_h$, we can then invert \eqref{charge} and \eqref{Jpsi} to express the parameters $\ell_h$ and $m_h$ purely in terms of the physical parameters, 
\begin{equation}
\ell_h = \ell_h(Q, q_i)\,,\qquad m_h = m_h(Q, J_\psi, q_i)\,,
\end{equation}
and using these it is then clear that \eqref{Jphi_new} implies the single constraint amongst the physical parameters is of the form
\begin{equation}
J_\phi = J_\phi (Q,J_\psi, q_i)\,.
\end{equation}
In the case of a soliton
solution, all the parameters $\ell_i$ and $m_i$ are completely determined by
the $k_i$ (and hence the $q_i$), which then implies the charge \eqref{charge}
and angular momenta (\ref{Jpsi},\ref{Jphi_new}) can be expressed solely in
terms of the magnetic potentials,
\begin{equation}
  Q = Q(q_i)\,, \qquad J_\psi = J_\psi(q_i)\,, \qquad J_\phi = J_\phi(q_i)\,,
\end{equation}
thus giving three constraints on the physical parameters as anticipated above.
\section{Discussion}
\label{sec:dis}

In this work we have presented a complete classification of asymptotically flat, supersymmetric and biaxisymmetric solutions to five-dimensional minimal supergravity, which are regular on and outside an event horizon. Our analysis also covers the case of spacetimes containing no black hole, in which case we obtain a complete classification of soliton spacetimes in this class. The essential local result is that such solutions {\it must} be in the class of multi-centred Gibbons--Hawking solutions. Although these have been extensively studied over the last decade or so, a global analysis of these solutions has not been previously presented and therefore our work also fills this important gap  \footnote{A global analysis of a subclass of supersymmetric solutions with a Gibbons--Hawking base which reduce to four-dimensional Euclidean Einstein--Maxwell solutions was performed in~\cite{Dunajski:2006vs}.}.  We reveal a rich moduli space of $n$-centred solutions both with and without a black hole.

One of the main global results is that we find a refinement of the allowed horizon topologies in this class. That is, horizon cross-sections must be $S^3$, $S^1\times S^2$ or a lens space $L(p,1)$, in particular ruling out $L(p,q)$ and $q \neq 1$ (mod $p$). Although examples of black hole solutions have been previously constructed  for each possible type, we find that there are an infinite number of distinct black hole solutions for each of the horizon topologies. More precisely, the number of distinct $n$-centred solutions containing a single black hole grows rapidly with $n$ (see equations (\ref{singleLp}) and (\ref{singleBH})).

An important technical problem which we were unable to solve is whether the constraints on the parameters of the solution required by smoothness of the horizon and the axes (given in Theorem \ref{classthm}) are in fact sufficient for smoothness and stably causality everywhere in the DOC. Based on numerical checks performed for the known examples we believe this is indeed the case, although this issue requires further investigation. Recently, progress in this direction for the bubbling soliton solutions was made~\cite{Avila:2017pwi}; it would be interesting to see if a similar method could be applied to the black hole case.

It is interesting to compare our results to vacuum gravity.  Here the classification of asymptotically flat, stationary and biaxisymmetric spacetimes is an open problem. It is known that black holes in this class must have horizons of $S^3$, $S^1\times S^2$ or lens space $L(p,q)$ topology, however, it is not known whether a smooth solution exists for every possible rod structure. Indeed, the only known explicit solutions are the $S^3$ Myers--Perry black hole and the $S^1\times S^2$ black ring, both of which have the simplest possible rod structure. Given that we have found a refinement of the allowed horizon topologies for supersymmetric black holes, it is interesting to consider if this also happens for vacuum black holes.   In fact we have shown that supersymmetry restricts the possible rod structures in such a way to constrain the horizon topology. In contrast, in the vacuum case, rod structures for black holes with $L(p,q)$ and $q \neq 1$ (mod $p$) are possible, although it is not known whether there exist corresponding {\it smooth} spacetimes~\cite{Khuri:2017xsc}. It is therefore still possible that such topologies are also not realised for vacuum black holes, although this remains an open problem.

On the other hand, if regular vacuum black holes with $L(p,q)$ and $q \neq 1$ (mod $p$) do exist, one then expects to be able to construct charged non-extremal versions of these in minimal supergravity and our results then show that such solutions would {\it not} have a supersymmetric limit.  This is not what occurs for the known families of spherical black holes and black rings, where the supersymmetric case always arises as a limit case of a larger non-extremal family.

There are a number of possible directions in which our work could be extended. Clearly, a similar classification in the more general minimal supergravity coupled to an arbitrary number of vector multiplets could be carried out, where one anticipates analogous results.  It would also be interesting to adapt our analysis to spacetimes with other relevant asymptotics such as Kaluza--Klein or Taub-NUT. Indeed, the local version of our horizon analysis, Theorem \ref{thm:horizon}, could be applied directly in these cases.  

It would be interesting to investigate the implications of our results for black hole non-uniqueness and the related problem of counting of black hole microstates in string theory.  Recently it was shown that a black hole in a spacetime with a single bubble in the DOC may have the same conserved charges as the standard spherical BMPV black hole (thereby demonstrating continuous violation of uniqueness even for spherical black holes)~\cite{Kunduri:2014iga}. Furthermore, it was also shown that this solution has higher entropy than the BMPV black hole as one approaches the BMPV upper spin limit~\cite{Horowitz:2017fyg}. Our classification presents the opportunity to analyse the {\it full} space of solutions with the same charges (and symmetry) as the standard solutions. We leave this interesting question to future work. \\

\noindent {\bf Acknowledgements}. VB is supported by an EPSRC studentship.  JL is supported by STFC [ST/L000458/1]. 

\bibliographystyle{jhep}
\bibliography{paper}

\end{document}